\documentclass[12pt]{article}
\usepackage{fullpage}

\usepackage{algorithm}
\usepackage{algpseudocode}
\usepackage{amsmath}
\usepackage{amssymb}
\usepackage{amsthm}
\usepackage{appendix}
\usepackage{acro}
\usepackage{bm}
\usepackage{booktabs} 
\usepackage{circledsteps}
\usepackage{enumitem}
\usepackage{graphicx}
\usepackage{hyperref}
\usepackage{cleveref}
\usepackage{microtype}
\usepackage{natbib}
\usepackage{physics}
\usepackage{placeins}
\usepackage{xcolor}
\usepackage{subcaption}
\usepackage{booktabs}
\usepackage[normalem]{ulem}

\usepackage{tikz}
\usepackage{longtable}
\usepackage{colortbl}
\usepackage{pgfplots}

\newtheorem{definition}{Definition}
\newtheorem{example}{Example}
\newtheorem{theorem}{Theorem}
\newtheorem{lemma}{Lemma}
\newtheorem{remark}{Remark}

\DeclareMathAlphabet\matheuvm{U}{zeur}{m}{n}

\title{Reinforcement Learning for Adaptive MCMC}

\author{Congye Wang$^1$, Wilson Chen$^2$, Heishiro Kanagawa$^1$, Chris. J. Oates$^{1,3}$ \\
\small $^1$ Newcastle University, UK \\
\small $^2$ University of Sydney, Australia \\
\small $^3$ Alan Turing Institute, UK }

\acsetup{format/first-long=\itshape}

\DeclareAcronym{pdf}{short=PDF, long=probability density function}
\DeclareAcronym{mala}{short=MALA, long=Metropolis-adjusted Langevin algorithm}
\DeclareAcronym{amala}{short=AMALA, long=adaptive Metropolis-adjusted Langevin algorithm}
\DeclareAcronym{mcmc}{short=MCMC, long=Markov chain Monte Carlo}
\DeclareAcronym{rl}{short=RL, long=reinforcement learning}
\DeclareAcronym{mdp}{short=MDP, long=Markov decision process}
\DeclareAcronym{ddpg}{short=DDPG, long=deep deterministic policy gradient}
\DeclareAcronym{esjd}{short=ESJD, long=expected squared jump distance}
\DeclareAcronym{ssage}{short=SSAGE, long=simultaneously strongly aperiodically geometrically ergodic}
\DeclareAcronym{rlmh}{short=RLMH, long=Reinforcement Learning Metropolis--Hastings}
\DeclareAcronym{arwmh}{short=ARWMH, long=adaptive random walk Metropolis--Hastings}
\DeclareAcronym{nuts}{short=NUTS, long=no U-turn sampler}
\DeclareAcronym{mmd}{short=MMD, long=maximum mean discrepancy}
\DeclareAcronym{ess}{short=ESS, long=effective sample size}

\DeclareMathOperator*{\argmin}{arg\,min}

\definecolor{rlcolor}{rgb}{0.41, 0.4, 0.54}
\definecolor{amcolor}{rgb}{0.72, 0.5, 0.44}
\definecolor{nutscolor}{rgb}{0.7, 0.69, 0.71}
\definecolor{malacolor}{rgb}{0.57, 0.68, 0.62}

\begin{document}

\maketitle

\begin{abstract}
An informal observation, made by several authors, is that the adaptive design of a Markov transition kernel has the flavour of a reinforcement learning task.
Yet, to-date it has remained unclear how to actually exploit modern reinforcement learning technologies for adaptive MCMC.
The aim of this paper is to set out a general framework, called \emph{Reinforcement Learning Metropolis--Hastings}, that is theoretically supported and empirically validated.
Our principal focus is on learning fast-mixing Metropolis--Hastings transition kernels, which we cast as deterministic policies and optimise via a policy gradient.
Control of the learning rate provably ensures conditions for ergodicity are satisfied.
The methodology is used to construct a gradient-free sampler that out-performs a popular gradient-free adaptive Metropolis--Hastings algorithm on $\approx 90 \%$ of tasks in the \texttt{PosteriorDB} benchmark. 
\end{abstract}

\section{Introduction}

A vast literature on algorithms, tips, and tricks is testament to the success of \ac{mcmc}, which remains the most popular approach to numerical approximation of probability distributions characterised up to an intractable normalisation constant.
Yet the breadth of methodology also presents a difficulty in selecting an appropriate algorithm for a specific task. 
The goal of \emph{adaptive} \ac{mcmc} is to automate, as much as possible, the design of a fast-mixing Markov transition kernel.
To achieve this, one alternates between observing the performance of the current transition kernel, and updating the transition kernel in a manner that is expected to improve its future performance \citep{andrieu2008tutorial}.
Though the online adaptation of a Markov transition kernel in principle sacrifices the ergodicy of  \ac{mcmc}, there are several ways to prove that ergodicity is in fact retained if the transition kernel converges fast enough (in an appropriate sense) to a sensible limit.

There is at least a superficial relationship between adaptive \ac{mcmc} and \ac{rl}, with both attempting to perform optimisation in a \ac{mdp} context. 
Several authors have noted this similarity, yet to-date it has remained unclear whether state-of-the-art \ac{rl} technologies can be directly exploited for adaptive \ac{mcmc}.
The demonstrated success of \ac{rl} in tackling diverse \acp{mdp}, including autonomous driving \citep{kiran2021deep}, gaming \citep{silver2016mastering}, and natural language processing \citep{shinn2024reflexion}, suggests there is a considerable untapped potential if \ac{rl} can be brought to bear on adaptive \ac{mcmc}.

This paper sets out a general framework in \Cref{sec: methods}, called \emph{Reinforcement Learning Metropolis--Hastings}, in which the parameters of Metropolis--Hastings transition kernels are iteratively optimised along the \ac{mcmc} sample path via a policy gradient.
In particular, we explore transition kernels parametrised by neural networks, and leverage state-of-the-art deterministic policy gradient algorithms from \ac{rl} to learn suitable parameters for the neural network.
Despite the apparent complexity of the set-up, at least compared to more standard methods in adaptive \ac{mcmc}, it is shown in \Cref{sec: theory} how control of the learning rate and gradient clipping can be used to provably guarantee that diminishing adaptation and containment conditions, which together imply ergodicity of the resulting Markov process, are satisfied.
The methodology was objectively stress-tested, with results on the \texttt{PosteriorDB} benchmark reported in \Cref{sec: empirical}.

\subsection{Related Work}

Before presenting our methodology, we first provide a brief summary of the extensive literature on adaptive \ac{mcmc} (\Cref{subsec: adapt mcmc}), and a comprehensive review of existing work at the interface of \ac{mcmc} and \ac{rl} (\Cref{subsec: rl meets mcmc}).
Let $\mathcal{X}$ be a topological space equipped with the Borel $\sigma$-algebra; henceforth the measurability of relevant functions and sets will always be assumed.
For notation, we mainly work with densities throughout and use $p(\cdot)$ denote the density of the target distribution, with respect to an appropriate reference measure on $\mathcal{X}$.

\subsubsection{Adaptive MCMC}
\label{subsec: adapt mcmc}

For the most part, research into adaptive \ac{mcmc} has focused on the adaptive design of a fast-mixing Metropolis--Hastings transition kernel \citep{haario2006dram,roberts2009examples}, though the adaptive design of other classes of \ac{mcmc} method, such as Hamiltonian Monte Carlo, have also been considered \citep{wang2013adaptive,hoffman2014no,christiansen2023self}.
Recall that Metropolis--Hastings refers to a Markov chain $(x_n)_{n \in \mathbb{N}} \subset \mathcal{X}$ such that, to generate $x_{n+1}$ from $x_n$, we first simulate a candidate state $x_{n+1}^\star \sim q(\cdot | x_n)$, where the collection $\{ q(\cdot | x) : x \in \mathcal{X}\}$ is called the \emph{proposal}, and then set $x_{n+1} = x_{n+1}^\star$ with probability
\begin{align}
    \alpha(x_n,x_{n+1}^\star) := \min\left\{ 1 , \frac{p(x_{n+1}^\star)}{p(x_n)} \frac{q(x_n | x_{n+1}^\star) }{ q(x_{n+1}^\star | x_n)} \right\} , \label{eq: std accept reject}
\end{align}
else we set $x_{n+1} = x_n$.
To develop an adaptive \ac{mcmc} method in this context three main ingredients are required:

\paragraph{Performance criterion}
The first ingredient is a criterion to be (approximately) optimised.
Standard choices include: the negative correlation between consecutive states $x_n$ and $x_{n+1}$, or between $x_n$ and $x_{n+l}$ for some lag $l$; the average return time to a given set; the expected squared jump distance \citep{pasarica2010adaptively}; or the asymptotic variance associated to a function $f(x_n)$ of interest \citep{andrieu2001controlled}.
For Metropolis--Hastings chains specifically, the average acceptance rate is a criterion that is widely-used, but alternatives include criteria based on raw acceptance probabilities \citep{titsias2019gradient}, and using a divergence between the proposal and the target distributions \citep{andrieu2006ergodicity,dharamshi2023sampling}.

\paragraph{Candidate transition kernels}
The second ingredient is a set of candidate transition kernels, among which an adaptive \ac{mcmc} method aims to identify an element that is optimal with respect to the performance criterion.
In the Metropolis--Hastings context, a popular approach is to use previous samples $( x_i )_{i=1}^n$ to construct a rough approximation $p_n$ to the target distribution $p$, and then to exploit $p_n$ in the construction of a Metropolis--Hastings proposal.
For example, $p_n$ might be a Gaussian approximation based on the mean and covariance of the states previously visited, and the proposal may be $p_n$ itself.
Several authors have proposed more flexible approximation methods, such as a mixture of $t$-distributions \citep{tran2016adaptive}, local polynomial approximation of the target \citep{conrad2016accelerating}, and normalising flows \citep{gabrie2022adaptive}.
Such methods can suffer from substantial autocorrelation during the warm-up phase, but once a good approximation has been learned, samples generated using this approach can be almost independent \citep{davis2022rate}.

\paragraph{Mechanism for adaptation}
The final ingredient is a mechanism to adaptively select a suitable transition kernel, based on the current sample path $(x_i)_{i=1}^n$, in such a manner that the ergodicity of the Markov chain is still assured.
General sufficient conditions for ergodicity of adaptive \ac{mcmc} have been obtained, and these can guide the construction of a mechanism for selecting a transition kernel \citep{atchade2005adaptive,andrieu2006ergodicity,atchade2011adaptive}.
Some of the more recent research directions in adaptive \ac{mcmc} include the use of Bayesian optimisation \citep{mahendran2012adaptive}, incorporating global jumps between modes as they are discovered \citep{pompe2020framework}, local adaptation of the step size in the \ac{mala} \citep{biron2023automala}, tuning gradient-based \ac{mcmc} by directly minimising a divergence between $p$ and the \ac{mcmc} output \citep{coullon2023efficient}, and the online training of diffusion models to approximate the target \citep{hunt2024accelerating}.
Our contribution will, in effect, demonstrate that state-of-the-art methods from \ac{rl} can be added to this list.

\subsubsection{Reinforcement Learning Meets MCMC}
\label{subsec: rl meets mcmc}

Consider a \ac{mdp} with a \textit{state} set $\mathcal{S}$, an \textit{action} set $\mathcal{A}$, and an \textit{environment} which determines both how the state $s_n$ is updated to $s_{n+1}$ in response to an action $a_n$, and the (scalar) \emph{reward} $r_n$ that resulted.
\textit{Reinforcement learning} refers to a broad class of methods that attempt to learn a \emph{policy} $\pi : \mathcal{S} \rightarrow \mathcal{A}$, meaning a mechanism to select actions $a_n = \pi(s_n)$, which aims to maximise the expected cumulative reward \citep{kaelbling1996reinforcement}.
A useful taxonomy of modern \ac{rl} algorithms is based on whether the policy $\pi$ is deterministic or stochastic, and whether the action space is discrete or continuous \citep{franccois2018introduction}.
Techniques from deep learning are widely used in modern \ac{rl}, for example in \emph{Deep $Q$-Learning} (for discrete actions) \citep{mnih2015human} and in \emph{Deep Deterministic Policy Gradient} (for continuous actions) \citep{lillicrap2015continuous}.
The impressive performance of modern \ac{rl} serves as strong motivation for investigating if and how techniques from \ac{rl} can be harnessed for adaptive \ac{mcmc}.
Several authors have speculated on this possibility, but to-date the problem remained unsolved:

\textit{Policy Guided Monte Carlo} is an approach developed to sample configurations from discrete lattice models, proposed in \citet{bojesen2018policy}.
The approach identifies the state $s_n$ with the state $x_n$ of the Markov chain,  the (stochastic) policy $\pi$ with the proposal distribution $q$ in Metropolis--Hastings \ac{mcmc}, and the (discrete) action $a_n$ with the candidate $x_{n+1}^\star$ that is being proposed.
At first sight this seems natural, but we contend that this set-up is not appropriate \ac{rl}, because the action $x_{n+1}^\star$ contains insufficient information to determine the acceptance probability for the proposed state $x_{n+1}^\star$; to achieve that, we would also need to know the ratio of probabilities $q(x_n | x_{n+1}^\star) / q(x_{n+1}^\star | x_n)$ appearing in \eqref{eq: std accept reject}.
A possible mitigation is to restrict attention to symmetric proposals, for which this ratio becomes identically 1, but symmetry places a limitation on the potential performance of \ac{mcmc}.
As a result, this set-up does not fit well with the usual formulation of \ac{rl}, and modern \ac{rl} methods cannot be directly deployed.

An alternative set-up is to again associate the state $s_n$ with the state $x_n$ of the Markov chain, but now associate the action set $\mathcal{A}$ with a set of Markov transition kernels, so that the policy determines which transition kernel to use to move from the current state to the next.
This approach was called \textit{Reinforcement Learning Accelerated MCMC} in \cite{chung2020multi}, where a finite set of Markov transition kernels were developed for a particular multi-scale inverse problem, and called \textit{Reinforcement Learning-Aided MCMC} in \cite{wang2021reinforcement}, where \ac{rl} was used to learn frequencies for updating the blocks of a random-scan Gibbs sampler.
The main challenge in extending this approach to general-purpose adaptive \ac{mcmc} is the difficulty in parametrising a collection of transition kernels in such a manner that the gradient-based policy search will work well; perhaps as a consequence, existing work did not leverage state-of-the-art techniques from \ac{rl}.
In addition, it appears difficult to establish conditions for ergodicity with this set-up, with these existing methods being only empirically validated.
Our contribution is to propose an alternative set-up, where we parametrise the \emph{proposal} in Metropolis--Hastings \ac{mcmc} instead of the transition kernel, showing how this enables both of these open challenges to be resolved.

\section{Methods}
\label{sec: methods}

Recall that we work with densities\footnote{This choice sacrifices generality and is not required for our methodology, but serves to make the presentation more transparent.} and associate $p(\cdot)$ with a probability measure on $\mathcal{X}$.
The task is to generate samples from $p(\cdot)$ using adaptive \ac{mcmc}.
Let $\Phi$ be a topological space, and for each $\varphi \in \Phi$ let $\{q_\varphi(\cdot | x) : x \in \mathcal{X}\}$ denote a proposal.
Simple proposals will serve as building blocks for constructing more sophisticated Markov transition kernels, as we explain next.

\subsection{$\phi$-Metropolis--Hastings}

Suppose now that we are given a map $\phi : \mathcal{X} \rightarrow \Phi$.
This mapping defines a Markov chain such that, if the current state is $x_n$, the next state $x_{n+1}$ is generated using a Metropolis--Hastings transition kernel with proposal $q_{\phi(x_n)}(\cdot | x_n)$.
In other words, we have a transition kernel
\begin{align*}
    x_{n+1}^\star \sim q_{\phi(x_n)}(\cdot | x_n), \qquad 
    x_{n+1} \gets \left\{ \begin{array}{ll} x_{n+1}^\star & \text{with probability } \alpha_\phi(x_n,x_{n+1}^\star) \\ x_n & \text{otherwise} \end{array} \right. 
\end{align*}
with a $\phi$-dependent acceptance probability 
\begin{align}
    \alpha_\phi(x_n,x_{n+1}^\star) := \min\left\{ 1 , \frac{p(x_{n+1}^\star)}{p(x_n)} \frac{q_{\phi(x_{n+1}^\star)}(x_n | x_{n+1}^\star) }{ q_{\phi(x_n)}(x_{n+1}^\star | x_n)} \right\} . \label{eq: new accept reject}
\end{align}
The design of a fast-mixing Markov chain can then be formulated as the design of a suitable map $\phi$, characterising the state-dependent proposal.
The Markov chain associated to $\phi$ will be called \textit{$\phi$-MH} in the sequel.
Two illustrative examples are presented, where in each case we suppose $p(\cdot)$ is a distribution with (for illustrative purposes, \emph{known}) mean $\mu$ and standard deviation $\sigma$ on $\mathcal{X} = \mathbb{R}$:

\begin{example}[$\phi$-MH based on symmetric random walk proposal]
    Consider the case where $q_\varphi(\cdot | x)$ is Gaussian with mean $x$ and standard deviation $\varphi$, corresponding to a \emph{symmetric random walk}.
    Then $\phi$-MH with $\phi(x) = \sigma + |x-\mu|$ is a Metropolis--Hastings chain that attempts to take larger steps when the chain is further away from the mean $\mu$ of the target.
\end{example}

\begin{example}[$\phi$-MH based on independence sampler proposal]
\label{ex: independence sampler}
    Consider the case where $q_\varphi(\cdot | x)$ is Gaussian with mean $\varphi$ and standard deviation $\sigma$, corresponding to an \emph{independence sampler}.
    Then $\phi$-MH with $\phi(x) = 2 \mu - x$ is a Metropolis--Hastings chain that induces anti-correlation by promoting jumps from one side of $\mu$ to the other.
\end{example}

These examples of $\phi$-MH can be analysed using existing techniques and their ergodicity can be established, but for general $\phi$ we will require some regularity to ensure $\phi$-MH generates samples from the correct target.
\Cref{lem: ergodic} below establishes weak conditions\footnote{The conclusion of \Cref{lem: ergodic} can be extended to general topological spaces $\mathcal{X}$, by noting that the arguments used in the proof are all topological, but for the present purposes such a generalisation was not pursued.} under which $\phi$-MH transition is $p$-invariant and \emph{ergodic}, the latter understood in this paper to mean convergence in total variation of $\mathrm{Law}(x_n)$ to $p(\cdot)$ is guaranteed from any initial state $x_0 \in \mathcal{X}$.

\begin{lemma}[Ergodicity of $\phi$-MH; $\mathcal{X} = \mathbb{R}^d$] \label{lem: ergodic}
    Let $\phi$ be continuous, and let both $x \mapsto p(x)$ and $(\varphi,x,y) \mapsto q_\varphi(x,y)$ be positive and continuous.
    Then $\phi$-MH is $p$-invariant and ergodic.
\end{lemma}

The proof is contained in \Cref{app: ergodic proof}.
Of course, a suitable map $\phi$ is unlikely to be known \emph{a priori} since $p(\cdot)$ is intractable, so instead a suitable $\phi$ will need to be \emph{learned}.
Armed with a guarantee of correctness for $\phi$-MH, we now seek to cast the learning of a suitable map $\phi$ as a problem that can be addressed using modern techniques from \ac{rl}.

\subsection{$\phi$-Metropolis--Hastings as Reinforcement Learning}

Our main methodological contribution is to set out a correct framework for casting adaptive \ac{mcmc} as \ac{rl}, and represents a fundamental departure from the earlier attempts described in \Cref{subsec: rl meets mcmc}.
To cast the learning of the map $\phi$ appearing in $\phi$-MH as an \ac{mdp}, the state set $\mathcal{S}$, action set $\mathcal{A}$, and the environment each need to be specified:

\paragraph{State set $\mathcal{S}$}  
The \emph{state} in our \ac{mdp} set-up is
\begin{align*}
    s_n = \left[ \begin{array}{c} x_n \\ x_{n+1}^\star \end{array} \right] \in  \mathcal{X} \times \mathcal{X} =: \mathcal{S} ,
\end{align*}
consisting of the current state $x_n$ of the Markov chain and the proposed state $x_{n+1}^\star$, which will either be accepted or rejected.

\paragraph{Action set $\mathcal{A}$} 
The \emph{action} set in our set-up is $\mathcal{A} = \Phi \times \Phi$.
A map $\phi : \mathcal{X} \rightarrow \Phi$ induces a (deterministic) policy of the form
\begin{align}
    \pi(s_n) = a_n = \left[ \begin{array}{c} \phi(x_n) \\ \phi(x_{n+1}^\star) \end{array} \right] \in \mathcal{A} . \label{eq: policy}
\end{align}
The motivation for this set-up is that the environment will need to compute not just the probability $q_{\phi(x_n)}(x_{n+1}^\star | x_n)$ of sampling the candidate $x_{n+1}^\star$, but also the reverse probability $q_{\phi(x_{n+1}^\star)}(x_n | x_{n+1}^\star)$, to calculate the overall acceptance probability in \eqref{eq: new accept reject}.

\paragraph{Environment}
The \emph{environment}, given the current state $s_n$ and action $a_n$, executes three tasks:
First, an accept/reject decision is made so that $x_{n+1} = x_{n+1}^\star$ with probability \eqref{eq: new accept reject}, else $x_{n+1} = x_n$.
Second, the environment simulates $\smash{x_{n+2}^\star \sim q_{\phi(x_{n+1})}(\cdot | x_{n+1})}$ and returns the updated state $\smash{s_{n+1} = [x_{n+1},x_{n+2}^\star]}$.
This is possible since $\phi(x_{n+1})$ is equal to either $\phi(x_n)$ or $\phi(x_{n+1}^\star)$, each of which were provided (via $a_n$) to the environment.
Third, the environment computes a reward $r_n$.
For the case $\mathcal{X} = \mathbb{R}^d$, we define our reward as
\begin{align*}
    r_n =2 \log \|x_n - x_{n+1}^\star \| + \log \alpha_\phi(x_n,x_{n+1}^\star) ,
\end{align*}
which is the logarithm of the \ac{esjd} from $x_n$ to $x_{n+1}$, where the expectation is computed with respect to the randomness in the accept/reject step.

\begin{remark}[Choice of action]
    Could we instead define the action to be $\phi$, so that at each iteration the policy picks a transition kernel?
    In principle yes, but then the action space would be high- or infinite-dimensional, severely increasing the difficulty of the \ac{rl} task.
    Our set-up allows us to flexibly parametrise $\phi$ using a neural network, while ensuring the number of parameters in this network has no bearing on the dimension of the action set. 
\end{remark}

\begin{remark}[Choice of reward]
    The \ac{esjd} is a popular criterion for use in adaptive \ac{mcmc}, due to its close relationship with mixing times \citep{sherlock2009optimal} and the ease with which it can be computed.
    Our initial investigations found the logarithm of \ac{esjd} to be more useful for \ac{rl}, as it provides the reward with a greater dynamic range to distinguish bad policies from very bad policies, leading to improved performance of methods based on policy gradient (see \Cref{subsec: policy gradient}).
\end{remark}

\subsection{Learning $\phi$ via Policy Gradient}
\label{subsec: policy gradient}

The rigorous formulation of $\phi$-MH as an \ac{mdp} enables modern techniques from \ac{rl} to be immediately brought to bear on adaptive \ac{mcmc}.
The objective that we aim to maximise is the expected reward at stationarity; $J(\phi) := \mathbb{E}_{\phi} [ r ]$, where the state $s$ is sampled from the stationary distribution of the \ac{mdp} and the action $a = \pi(s)$ is determined by the $\phi$-MH policy $\pi$ in \eqref{eq: policy}.
In practice one restricts attention to a parametric family $\phi_\theta$ for some $\theta \in \mathbb{R}^p$, $p \in \mathbb{N}$.
To optimise $J(\phi_\theta)$ we employ a (deterministic) \emph{policy gradient} method \citep{silver2014deterministic}, meaning in our case that we alternate between updating the Markov chain from $x_n$ to $x_{n+1}$ using $\phi_\theta$-MH with $\theta = \theta_n$ fixed, and updating the parameter $\theta$ using approximate gradient ascent 
\begin{align}
    \theta_{n+1} \gets \theta_n + \alpha_n \left. \nabla_\theta J_n(\phi_\theta) \right|_{\theta = \theta_n}  \label{eq: grad ascent}
\end{align}
where $\alpha_n \geq 0$ is called a \emph{learning rate}. 
Here $\nabla_\theta J_n$ indicates an approximation to the \emph{policy gradient} $\nabla_\theta J$, which may be constructed using any of the random variables generated up to that point.
A considerable amount of research effort in \ac{rl} has been devoted to approximating the policy gradient, and for the experiments reported in \Cref{sec: empirical} we used the \ac{ddpg} method of \citet{lillicrap2015continuous}.
\ac{ddpg} is suitable for deterministic policies and a continuous action set, and operates by training a \emph{critic} based on experience stored in a \emph{replay buffer} to enable a stable approximation to the policy gradient.
Since the details of \ac{ddpg} are not a novel contribution of our work, we reserve them for \Cref{app: implementation}.

The proposed \ac{rlmh} method is stated in \Cref{alg: rlmh}, and is compatible with \emph{any} approach to approximation of the policy gradient, not just \ac{ddpg}.
Indeed, the theoretical results that we present in \Cref{sec: theory} leverage the summability of the learning rate sequence $(\alpha_n)_{n \geq 0}$ and norm-based \emph{gradient clipping} in \Cref{alg: rlmh} to prove that sufficient conditions for ergodicity are satisfied.

\begin{algorithm}
\caption{Reinforcement Learning Metropolis--Hastings (RLMH)}
\label{alg: rlmh}
\begin{algorithmic}
\Require $x_0 \in \mathcal{X}$, $\theta_0 \in \mathbb{R}^p$, gradient threshold $\tau > 0$, learning rate $(\alpha_n)_{n \geq 0} \subset [0,\infty)$
\Ensure {\footnotesize $\sum_{n \geq 0}$} $\alpha_n < \infty$
\For{$n = 0,1,2,\dots $}
    \State \smash{$x_{n+1}^\star \sim q_{\phi_{\theta_n}(x_n)}(\cdot | x_n)$} \Comment{propose next state}
    \State \smash{$x_{n+1} \gets x_{n+1}^\star$ with probability $\alpha_{\phi_{\theta_n}}(x_n,x_{n+1}^\star)$, else $x_{n+1} \gets x_n$} \Comment{accept/reject}
    \State \smash{$g_n \gets \nabla_\theta J_n(\phi_\theta) |_{\theta = \theta_n}$} \Comment{approximate policy gradient}
    \State \textbf{if} \; $\|g_n\| > \tau$ \; \textbf{then} \; $g_n \gets \tau g_n / \|g_n\|$ 
        \Comment{gradient clipping}
    \State $\theta_{n+1} \gets \theta_n + \alpha_n g_n$  \Comment{policy update}
\EndFor
\end{algorithmic}
\end{algorithm}

\begin{remark}[On-policy requirement]
\label{rem: on policy}
    Our set-up is \emph{on-policy}, meaning that only actions specified by the policy are used to evolve the Markov chain; this is necessary for maintaining detailed balance in $\phi$-MH, which in turn ensures the Markov transition kernels are $p$-invariant.
    On the other hand, off-policy exploration can be concurrently conducted to aid in approximating the policy gradient \citep[for example, as training data for the critic in \ac{ddpg};][]{ladosz2022exploration}.
\end{remark}

\section{Theoretical Assessment}
\label{sec: theory}

Such is the generality of $\phi$-MH that it is possible to develop diverse gradient-free and gradient-based adaptive \ac{mcmc} algorithms using \ac{rlmh}.
Indeed, the building block proposals for $\phi$-MH could range from simple proposals (e.g. random walks) to complicated proposals (e.g. involving higher-order gradients of the target).
This section establishes sufficient conditions for the ergodicity of \ac{rlmh}, and to this end it is necessary to be specific about what building block proposals will be employed.
The remainder of this paper develops a \emph{gradient-free} sampling scheme in detail; our motivation is a tractable end-to-end theoretical analysis, guaranteeing correctness of the proposed method.

An accessible introduction to the theory of adaptive \ac{mcmc} is provided in \citet{andrieu2006ergodicity}.
At a high-level, the parameter $\theta$ of a Markov transition kernel is being allowed to depend on the sample path, i.e. $\theta_n \equiv \theta_n(x_0,\dots,x_n)$; intuitively, the hope is that $\theta_n$ will converge to a fixed limiting value $\theta_\star$, and that ergodicity properties of the adaptive chain will be inherited from those of the chain associated with $\theta_\star$.
However, carried out na\"{i}vely, the sequence $\theta_n$ can fail to converge, or converge to a value $\theta_\star$ for which the chain is no longer ergodic; see Section 2 in \citet{andrieu2008tutorial}.
In brief, the first issue can be solved by \emph{diminishing adaptation}; ensuring that the differences between $\theta_n$ and $\theta_{n+1}$ are decreasing, or even zero after a finite number of iterations have been performed.
The second issue can be resolved by \emph{containment}; restricting $\theta$ to e.g. compact subsets for which all transition kernels are well-behaved.

Our strategy below is to force diminishing adaptation via control of the learning rate and gradient clipping, and then to demonstrate containment via careful parametrisation of the $\phi$-MH Markov transition kernel.
The particular form of containment that we consider is motivated by the analytical framework of \citet{roberts2007coupling}, and requires us to establish that $\phi$-MH is \ac{ssage}; this condition is precisely stated in \Cref{def: ssge}.
For notation, recall that the \emph{transition kernel} $P(x,\cdot)$ of a Markov chain specifies the distribution of $x_{n+1}$ if the current state is $x_n = x$. 
For a function $f : \mathcal{X} \rightarrow \mathbb{R}$ we write $(Pf)(x) = \int f(y) P(x,\mathrm{d}y)$.
Let $\mathrm{1}_S$ denote the indicator function associated to a set $S$.

\begin{definition}[SSAGE; \citealp{roberts2007coupling}]
\label{def: ssge}
    A family of $p$-invariant Markov transition kernels $\{P_\theta\}_{\theta \in \Theta}$ is \emph{\ac{ssage}} if there exists $S \subset \mathcal{X}$, $V : \mathcal{X} \rightarrow [1,\infty)$, $\delta > 0$, $\lambda < 1$, and $b < \infty$, such that $\sup_{x \in S} V(x) < \infty$ and the following hold:
\begin{enumerate}
\itemsep0em 
    \item \emph{(minorisation)} For each $\theta \in \Theta$, there exists a probability measure $\nu_\theta$ on $S$ with $P_\theta(x,\cdot) \geq \delta \nu_\theta(\cdot)$ for all $x \in S$.
    \item \emph{(simultaneous drift)} $(P_\theta V)(x) \leq \lambda V(x) + b \mathrm{1}_S(x)$ for all $x \in \mathcal{X}$ and $\theta \in \Theta$.
\end{enumerate}
\end{definition}

To the best of our knowledge, existing \ac{ssage} results for Metropolis--Hastings focused on the case of a random walk proposal \citep{roberts1996geometric,roberts2007coupling,jarner2000geometric,bai2009containment,saksman2010ergodicity}.
However, the the proposal of $\phi$-MH can in principle be (much) more sophisticated than a random walk.
Our first contribution is therefore to generalise existing sufficient conditions for \ac{ssage} beyond the case of a random walk proposal, and in this sense \Cref{lem: ssge for phi mh} below may be of independent interest.

Let $\smash{\mathcal{P}(\mathbb{R}^d)}$ denote the set of probability distributions $p(\cdot)$ on $\mathbb{R}^d$ that are positive, continuously differentiable, and \emph{sub-exponential}, i.e.
\begin{align}
\lim_{\|x\| \rightarrow \infty} n(x) \cdot (\nabla \log p)(x) = - \infty  \label{eq: def subexponential}
\end{align}
where $\smash{n(x) := x / \|x\|}$.
To simplify notation, we use $\smash{q_\theta(\cdot | x)}$ as shorthand for $\smash{q_{\phi_\theta(x)}(\cdot|x)}$ and $\smash{\alpha_\theta(x,y)}$ as shorthand for $\smash{\alpha_{\phi_\theta}(x,y)}$.
Let $\smash{A_\theta(x) = \{ y \in \mathbb{R}^d : \alpha_\theta(x,y) = 1 \}}$ denote the region where proposals are always accepted, and let $\smash{\partial A_\theta^\delta(x) := \{y + s n(y) : y \in \partial A_\theta(x), |s| \leq \delta \}}$ denote a tube of radial width $\delta > 0$ around the boundary $\smash{\partial A_\theta(x)}$ of $A_\theta(x)$.

\begin{theorem}[SSAGE for general Metropolis--Hastings]
\label{lem: ssge for phi mh}
    Let $p \in \mathcal{P}(\mathbb{R}^d)$.
    Consider a family of $p$-invariant Metropolis--Hastings transition kernels $\{P_\theta\}_{\theta \in \Theta}$, where $P_\theta$ corresponds to a proposal $\{q_\theta(\cdot|x) : x \in \mathbb{R}^d\}$, and denote $Q_\theta(S|x) := \int_S q_\theta(y|x) \mathrm{d}y$ for $S \subseteq \mathbb{R}^d$.
    Let $(\theta,x,y) \mapsto q_\theta(y|x)$ be positive and continuous, and assume further that:

    \vspace{-10pt}
    \begin{enumerate}
    \itemsep0em 
        \item \emph{(quasi-symmetry)} $\rho := \sup_{x,y \in \mathbb{R}^d,  \theta \in \Theta} q_\theta(y|x) / q_\theta(x|y) < \infty$. 
        \item \emph{(regular acceptance boundary)} For all $\epsilon > 0$ there exists $\delta >0$ and $R > 0$ such that, for all $\|x\| \geq R$ and $\theta \in \Theta$, we have $Q_\theta( \partial A_\theta^\delta(x) | x ) < \epsilon$. 
        \item \emph{(minimum performance level)} $\liminf_{\|x\| \rightarrow \infty} \inf_{\theta \in \Theta} Q_\theta( A_\theta(x) | x ) > 0$. 
    \end{enumerate}

    \vspace{-10pt}
    \noindent Then $\{P_\theta\}_{\theta \in \Theta}$ is SSAGE.
\end{theorem}

The proof is contained in \Cref{app: general proof}, and follows the general strategy used to establish \ac{ssage} for random walk proposals in \citet{jarner2000geometric}, but with additional technical work to relax the random walk requirement (which corresponds to the case $\rho = 1$).

Quasi-symmetry imposes a non-trivial constraint on the form of the $\phi$-MH chains that can be analysed.
For both our theoretical and empirical assessment we focus on a setting that can be considered a multi-dimensional extension of \Cref{ex: independence sampler}, where the mean $\varphi$ of the proposal is specified as the output of a map $\phi_\theta(\cdot)$ with parameters $\theta \in \mathbb{R}^p$.
To be precise, let the set of $d \times d$ symmetric positive-definite matrices be denoted $\smash{\mathrm{S}_d^+}$ and, for $\smash{\Sigma \in \mathrm{S}_d^+}$, let $\smash{\|x\|_{1,\Sigma} := \|\Sigma^{-1/2}x\|_1}$ denote an associated norm on $\mathbb{R}^d$.
Then, for the remainder, we consider a Laplace proposal 
\begin{align}
    q_\theta(y|x) \propto  \exp\left( -  \|y - \phi_\theta(x)\|_{1,\Sigma} \right) , \label{eq: Laplace proposal}
\end{align}
for which sufficient conditions for ergodicity, including quasi-symmetry, can be shown to hold under appropriate assumptions on $\{\phi_\theta\}_{\theta \in \mathbb{R}^p}$.
This is the content of \Cref{thm: explicit}, which we present next.

To set notation, recall that a function $f : \mathbb{R}^p \rightarrow \mathbb{R}$ is said to be \emph{locally bounded} if, for all $\theta \in \mathbb{R}^p$, there is an open neighbourhood of $\theta$ on which $f$ is bounded.
Such a function is said to be \emph{locally Lipschitz} at $\theta \in \mathbb{R}^p$ if 
$$
\mathrm{LocLip}_\theta(f) := \limsup_{\vartheta \rightarrow \theta} \frac{|f(\vartheta) - f(\theta)|}{\|\vartheta - \theta\|} < \infty ,
$$
and, when this is the case, $\mathrm{LocLip}_\theta(f)$ is called the \emph{local Lipschitz constant} of $f$ at $\theta \in \mathbb{R}^p$.
Such a function is called \emph{Lipschitz} if $\smash{\mathrm{Lip}(f) := \sup_{\theta} \mathrm{LocLip}_\theta(f) < \infty}$ and, when this is the case, $\mathrm{Lip}(f)$ is called the \emph{Lipschitz constant}.
Let $\smash{\mathcal{P}_0(\mathbb{R}^d) \subset \mathcal{P}(\mathbb{R}^d)}$ denote the subset of distributions for which the \emph{interior cone condition} 
$$
\limsup_{\|x\| \rightarrow \infty} \; n(x) \cdot \frac{ (\nabla p)(x) }{ \|(\nabla p)(x)\| } < 0 
$$ 
is also satisfied.

\begin{theorem}[Ergodicity of \ac{rlmh}]
    \label{thm: explicit}
    Let $p \in \mathcal{P}_0(\mathbb{R}^d)$, $\Sigma \in \mathrm{S}_d^+$, $\Theta = \mathbb{R}^p$.
    Consider $\phi$-MH with proposal \eqref{eq: Laplace proposal}.
    Assume that each $x \mapsto \phi_\theta(x)$ is Lipschitz and $\sup_{x \in \mathbb{R}^d} \|x - \phi_\theta(x) \| < \infty$.
    Further assume the parametrisation of $\phi_\theta$ in terms of $\theta$ is \emph{regular}, meaning that the following maps are well-defined and locally bounded:
    (i) $\theta \mapsto \sup_{x \in \mathbb{R}^d} \|x - \phi_\theta(x) \|$; 
    (ii) $\theta \mapsto \mathrm{Lip}(x \mapsto \phi_\theta(x))$;
    (iii) $\theta \mapsto \sup_{x \in \mathbb{R}^d} \mathrm{LocLip}_\theta(\vartheta \mapsto \phi_\vartheta(x))$.
    Then \ac{rlmh} in \Cref{alg: rlmh} is $p$-invariant and ergodic, irrespective of the approach used to approximate the policy gradient.
\end{theorem}

The proof is contained in \Cref{app: proof of explicit theorem}, where containment is established using \Cref{lem: ssge for phi mh} and diminishing adaptation follows from control of the learning rate and gradient clipping, as in \Cref{alg: rlmh}.
The conditions (i)-(iii) in \Cref{thm: explicit} are satisfied in the experiments reported in \Cref{sec: empirical} through careful choice of the family of maps $\{\phi_\theta\}_{\theta \in \mathbb{R}^p}$; full details are reserved for \Cref{app: parametrisation}.

\section{Empirical Assessment}
\label{sec: empirical}

The aim of this section is to illustrate the behaviour of the gradient-free \ac{rlmh} algorithm that was analysed in \Cref{sec: theory}, and to compare performance against an established and popular gradient-free adaptive Metropolis--Hastings algorithm using a community benchmark.

\paragraph{Implementation of \ac{rlmh}}
An established adaptive sampling scheme was used to \emph{warm start} \ac{rlmh}.
For the purposes of this paper, we first run $m = 10^4$ iterations $(x_i)_{i=-m+1}^0$ of a popular \ac{arwmh} algorithm based on a gradient-free random walk proposal, and let $\Sigma$ denote an estimate of the covariance of $p(\cdot)$ so-obtained; full details are contained in \Cref{app: adaptive mcmc}.
The matrix $\Sigma$ is then used to define the norm appearing in \eqref{eq: Laplace proposal}, and the final iteration $x_0$ is used to initialise \ac{rlmh}.
To control for the benefit of a warm start, we use \ac{arwmh} as a comparator in the subsequent empirical assessment.
The proposal mean $\phi_\theta$ was taken as a neural network designed to satisfy the assumptions of \Cref{thm: explicit}; see \Cref{app: parametrisation}.
Sensitivity to the neural architecture was explored in \Cref{app: choice of network}.
The policy gradient was approximated using \ac{ddpg} with default settings as described in \Cref{subsec: training detail}.
Code to reproduce our experiments is available at \url{https://github.com/congyewang/Reinforcement-Learning-for-Adaptive-MCMC}.

\paragraph{Illustration:  Learning a Global Proposal}


    

\begin{figure}
    \centering
    \includegraphics[width=\linewidth,clip,trim = 2.8cm 0.2cm 2.7cm 0.6cm]{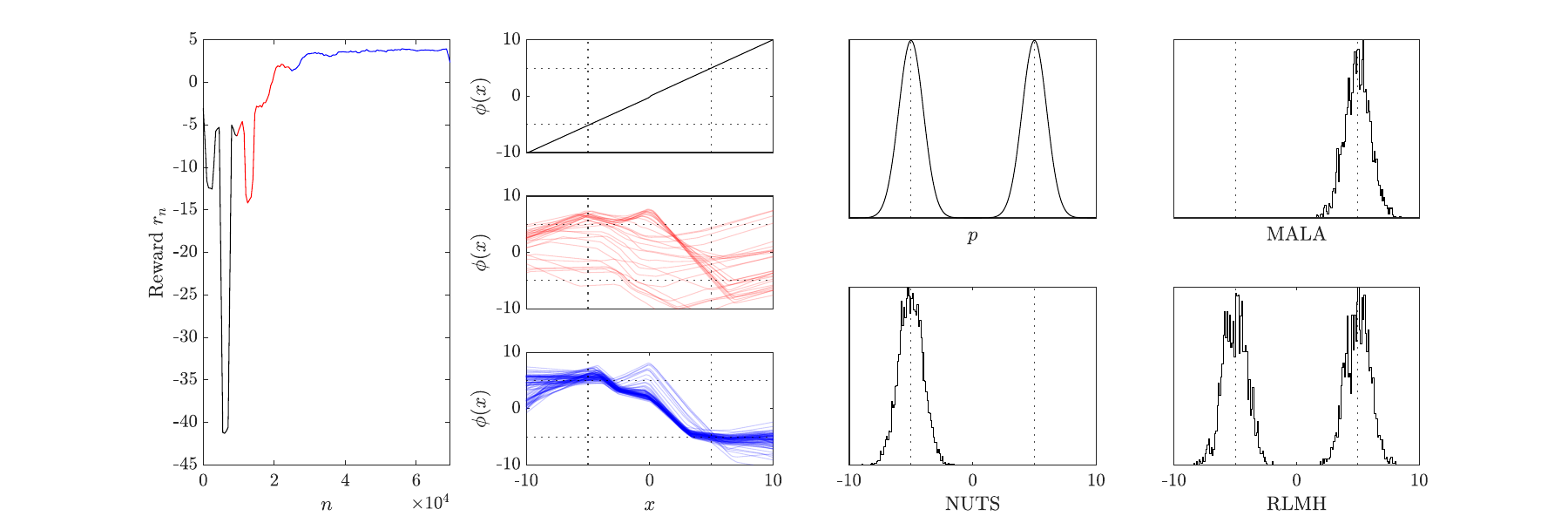}
    \tikz[overlay] \node at (0,5.4) {\scriptsize \Circled{0}};
    \tikz[overlay] \node at (-0.15,4.3) {\scriptsize \Circled{1}};
    \tikz[overlay] \node at (-0.3,2.3 ) {\scriptsize \Circled{2}};
    \tikz[overlay] \node at (-6.8,3.8) {\scriptsize \Circled{1}};
    \tikz[overlay] \node at (-5.5,5.5) {\scriptsize \Circled{2}};
    \tikz[overlay] \draw[|-|] (-7.3,4.2) -- (-6.8,4.2);
    \tikz[overlay] \draw[|-|] (-6.7,5.9) -- (-4.8,5.9);
    \caption{\acf{rlmh}, illustrated.  Here the task is to sample from the Gaussian mixture model $p(\cdot)$ whose equally-weighted components are $\mathcal{N}(\pm 5,1)$.
    Left: The reward sequence $(r_n)_{n \geq 0}$, where $r_n$ is the logarithm of the expected squared jump distance corresponding to iteration $n$ of \ac{rlmh}.
    Middle:  Proposal mean functions $x \mapsto \phi(x)$, at initialisation in \smash{\Circled{0}}, and corresponding to the rewards indicated in \smash{\Circled{1}} and \smash{\Circled{2}}.
    Right:  The density $p(\cdot)$, and histograms of the last $n = 5,000$ samples produced using \ac{mala}, \ac{nuts}, and \ac{rlmh}.
    [A smoothing window of length 5 was applied to the reward sequence to improve clarity of this plot.]
    \label{fig: illustration}
}
\end{figure}
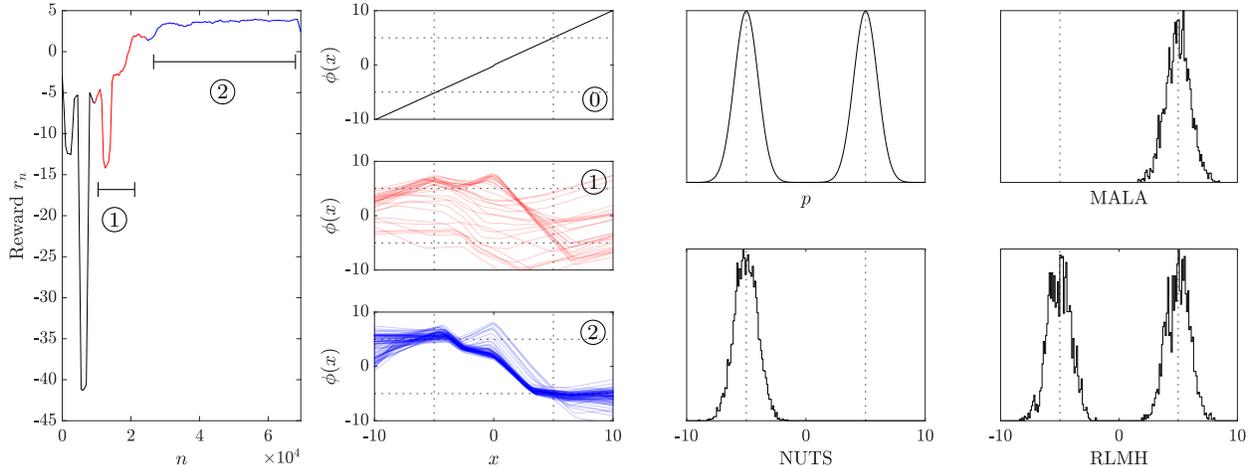

To understand the behaviour of \ac{rlmh}, consider the (toy) task of sampling from a two-component Gaussian mixture model in dimension $d = 1$ as presented in \Cref{fig: illustration}.
The proposal mean was initialised such that $\phi_\theta(x) \approx x$, corresponding to a random walk; this is visually represented in panel \smash{\Circled{0}}.
The sequence of rewards $(r_n)_{n \geq 0}$ obtained during training is displayed in the left panel.
Considerable exploration is observed in the period from initialisation to $n = 10^4$, while between $10^4$ and $2.5 \times 10^4$ the policy gradient directs the algorithm to an effective policy; see panel \smash{\Circled{1}}.
From then onward, the policy appears to have essentially converged; see panel \smash{\Circled{2}}.
The final learned policy performs \emph{global mode-hopping}; if the chain is currently in the effective support of one mixture component, then it will propose to move to the other mixture component.
This behaviour is consistent with seeking a large \ac{esjd}.
The samples generated from \ac{rlmh} are displayed as a histogram in the right hand panel of \Cref{fig: illustration}, and are seen to form a good approximation of the Gaussian mixture target.
These can be visually compared against samples generated using \ac{mala} \citep{robert1996exponential} and the \ac{nuts} \citep{hoffman2014no}, each of which struggled to escape from the mixture component where the Markov chain was initialised.
Further illustrations of \ac{rlmh}, which can be visualised in dimensions $d \in \{1,2\}$, can be found in \Cref{app: additional illus}.

\paragraph{Benchmarking on PosteriorDB}

\begin{table}[t!]
  \centering
  \scalebox{0.66}{
    \begin{tabular}{ccccccccccc}
    \toprule
          &       & \multicolumn{2}{c}{RLMH} & \multicolumn{2}{c}{ARWMH} & \multicolumn{2}{c}{AMALA} \\
\cmidrule{3-8}    Task  & $d$     & ESJD  & MMD   & ESJD  & MMD   & ESJD  & MMD \\
    \midrule
    \rowcolor{gray!20} earnings-earn\_height & 3     & \textbf{4.5(0.6)E3} & \textbf{1.8(0.1)E-1} & 2.1(0.0)E3 & 1.5(0.0)E0 & 1.3(0.9)E3 & 1.8(0.3)E0 \\
    \rowcolor{gray!20} earnings-log10earn\_height & 3     & \textbf{1.4(0.0)E-1} & \textbf{1.6(0.0)E-1} & 4.4(0.0)E-2 & 1.4(0.0)E0 & 1.4(0.0)E-1 & 1.6(0.0)E-1 \\
    \rowcolor{gray!20} earnings-logearn\_height & 3     & \textbf{3.2(0.0)E-1} & \textbf{1.6(0.0)E-1} & 1.0(0.0)E-1 & 1.5(0.0)E0 & 3.3(0.0)E-1 & 1.7(0.0)E-1 \\
    \rowcolor{gray!20} gp\_pois\_regr-gp\_regr & 3     & \textbf{3.7(0.1)E-1} & \textbf{1.2(0.0)E-1} & 1.2(0.0)E-1 & 1.1(0.0)E0 & 3.6(0.0)E-1 & 1.2(0.0)E-1 \\
    kidiq-kidscore\_momhs & 3     & \textbf{1.3(0.2)E0} & \textbf{1.5(0.0)E-1} & 7.2(0.0)E-1 & 1.3(0.0)E0 & 2.3(0.0)E0 & 1.4(0.0)E-1 \\
    kidiq-kidscore\_momiq & 3     & \textbf{3.6(0.2)E0} & \textbf{1.7(0.0)E-1} & 1.3(0.0)E0 & 1.5(0.0)E0 & 4.2(0.0)E0 & 1.6(0.0)E-1 \\
    \rowcolor{gray!20} kilpisjarvi\_mod-kilpisjarvi & 3     & \textbf{1.3(0.1)E1} & \textbf{1.7(0.0)E-1} & 6.5(0.0)E0 & 1.5(0.0)E0 & 6.1(3.1)E0 & 1.2(0.2)E0 \\
    \rowcolor{gray!20} mesquite-logmesquite\_logvolume & 3     & \textbf{1.2(0.0)E-1} & \textbf{1.3(0.0)E-1} & 3.7(0.0)E-2 & 1.1(0.0)E0 & 1.1(0.0)E-1 & 1.3(0.0)E-1 \\
    \rowcolor{gray!20} arma-arma11 & 4     & \textbf{6.4(0.0)E-2} & \textbf{1.2(0.0)E-1} & 1.9(0.0)E-2 & 1.1(0.0)E0 & 3.7(1.0)E-2 & 9.2(3.3)E-1 \\
    \rowcolor{gray!20} earnings-logearn\_height\_male & 4     & \textbf{4.1(0.1)E-1} & \textbf{1.6(0.0)E-1} & 1.2(0.0)E-1 & 1.5(0.0)E0 & 4.2(0.1)E-1 & 1.6(0.0)E-1 \\
    \rowcolor{gray!20} earnings-logearn\_logheight\_male & 4     & \textbf{1.7(0.1)E0} & \textbf{1.6(0.0)E-1} & 5.3(0.0)E-1 & 1.5(0.0)E0 & 1.8(0.0)E0 & 1.6(0.0)E-1 \\
    \rowcolor{gray!20} garch-garch11 & 4     & \textbf{8.0(0.2)E-1} & \textbf{1.4(0.0)E-1} & 2.8(0.0)E-1 & 1.2(0.0)E0 & 7.1(0.1)E-1 & 1.4(0.0)E-1 \\
    \rowcolor{gray!20} hmm\_example-hmm\_example & 4     & \textbf{4.6(0.1)E-1} & \textbf{1.3(0.0)E-1} & 1.4(0.0)E-1 & 1.2(0.0)E0 & 4.6(0.1)E-1 & 1.3(0.0)E-1 \\
    \rowcolor{gray!20} kidiq-kidscore\_momhsiq & 4     & \textbf{2.7(0.2)E0} & \textbf{1.4(0.0)E-1} & 1.3(0.0)E0 & 1.3(0.0)E0 & 4.4(0.1)E0 & 1.4(0.0)E-1 \\
    \rowcolor{gray!20} earnings-logearn\_interaction & 5     & \textbf{8.8(0.3)E-1} & \textbf{1.4(0.0)E-1} & 2.9(0.0)E-1 & 1.2(0.0)E0 & 1.1(0.0)E0 & 1.4(0.0)E-1 \\
    \rowcolor{gray!20} earnings-logearn\_interaction\_z & 5     & \textbf{8.7(0.1)E-2} & \textbf{1.2(0.0)E-1} & 2.7(0.0)E-2 & 1.1(0.0)E0 & 9.1(0.1)E-2 & 1.2(0.0)E-1 \\
    kidiq-kidscore\_interaction & 5     & \textbf{5.5(0.5)E0} & \textbf{1.7(0.1)E-1} & 3.8(0.0)E0 & 1.3(0.0)E0 & 1.4(0.0)E1 & 1.4(0.0)E-1 \\
    kidiq\_with\_mom\_work-kidscore\_interaction\_c & 5     & \textbf{7.2(0.5)E-1} & \textbf{1.6(0.0)E-1} & 5.2(0.0)E-1 & 1.3(0.0)E0 & 1.8(0.0)E0 & 1.3(0.0)E-1 \\
    kidiq\_with\_mom\_work-kidscore\_interaction\_c2 & 5     & \textbf{7.3(0.6)E-1} & \textbf{1.7(0.0)E-1} & 5.3(0.0)E-1 & 1.3(0.0)E0 & 1.9(0.0)E0 & 1.4(0.0)E-1 \\
    kidiq\_with\_mom\_work-kidscore\_interaction\_z & 5     & \textbf{1.0(0.1)E0} & \textbf{1.5(0.1)E-1} & 1.0(0.0)E0 & 1.1(0.0)E0 & 3.5(0.0)E0 & 1.2(0.0)E-1 \\
    kidiq\_with\_mom\_work-kidscore\_mom\_work & 5     & 9.6(1.3)E-1 & \textbf{1.8(0.1)E-1} & \textbf{1.2(0.0)E0} & 1.1(0.0)E0 & 4.2(0.0)E0 & 1.2(0.0)E-1 \\
    \rowcolor{gray!20} low\_dim\_gauss\_mix-low\_dim\_gauss\_mix & 5     & \textbf{6.7(0.0)E-2} & \textbf{1.1(0.0)E-1} & 2.1(0.0)E-2 & 9.9(0.0)E-1 & 7.0(0.1)E-2 & 1.1(0.0)E-1 \\
    \rowcolor{gray!20} mesquite-logmesquite\_logva & 5     & \textbf{2.5(0.0)E-1} & \textbf{1.2(0.0)E-1} & 7.5(0.0)E-2 & 1.1(0.0)E0 & 2.6(0.0)E-1 & 1.2(0.0)E-1 \\
    bball\_drive\_event\_0-hmm\_drive\_0 & 6     & \textbf{4.6(0.5)E-1} & \textbf{1.6(0.3)E-1} & 1.8(0.0)E-1 & 1.1(0.0)E0 & 6.2(0.3)E-1 & 1.4(0.1)E-1 \\
    sblrc-blr & 6     & \textbf{4.2(0.1)E-2} & \textbf{1.7(0.0)E-1} & 1.4(0.0)E-2 & 1.5(0.0)E0 & 4.6(0.0)E-2 & 1.6(0.0)E-1 \\
    sblri-blr & 6     & \textbf{4.2(0.1)E-2} & \textbf{1.7(0.0)E-1} & 1.3(0.0)E-2 & 1.5(0.0)E0 & 4.5(0.1)E-2 & 1.6(0.0)E-1 \\
    \rowcolor{gray!20} arK-arK & 7     & \textbf{1.2(0.0)E-1} & \textbf{1.1(0.0)E-1} & 3.5(0.0)E-2 & 9.5(0.0)E-1 & 1.4(0.0)E-1 & 1.1(0.0)E-1 \\
    \rowcolor{gray!20} mesquite-logmesquite\_logvash & 7     & \textbf{3.6(0.1)E-1} & \textbf{1.1(0.0)E-1} & 1.1(0.0)E-1 & 9.9(0.0)E-1 & 4.1(0.0)E-1 & 1.1(0.0)E-1 \\
    \rowcolor{gray!20} mesquite-logmesquite & 8     & \textbf{3.3(0.0)E-1} & \textbf{1.1(0.0)E-1} & 1.1(0.0)E-1 & 9.5(0.0)E-1 & 4.1(0.1)E-1 & 1.1(0.0)E-1 \\
    \rowcolor{gray!20} mesquite-logmesquite\_logvas & 8     & \textbf{3.4(0.1)E-1} & \textbf{1.1(0.0)E-1} & 1.1(0.0)E-1 & 9.6(0.0)E-1 & 4.1(0.0)E-1 & 1.1(0.0)E-1 \\
    mesquite-mesquite & 8     & 2.0(0.4)E1 & \textbf{5.1(1.1)E-1} & \textbf{6.5(0.0)E1} & 9.6(0.0)E-1 & 2.5(0.0)E2 & 1.1(0.0)E-1 \\
    eight\_schools-eight\_schools\_centered & 10    & 2.3(0.5)E-1 & \textbf{7.7(1.2)E-1} & \textbf{1.4(0.0)E0} & 1.1(0.0)E0 & 4.1(0.3)E0 & 1.5(0.2)E-1 \\
    \rowcolor{gray!20} eight\_schools-eight\_schools\_noncentered & 10    & \textbf{8.0(0.5)E-1} & \textbf{1.2(0.0)E0} & 6.2(0.0)E-1 & 1.2(0.0)E0 & 2.5(0.1)E0 & 1.2(0.0)E0 \\
    nes1972-nes & 10    & \textbf{2.4(0.0)E-1} & \textbf{1.1(0.0)E-1} & 9.2(0.0)E-2 & 9.6(0.0)E-1 & 3.6(0.0)E-1 & 1.0(0.0)E-1 \\
    \rowcolor{gray!20} nes1976-nes & 10    & \textbf{2.5(0.0)E-1} & \textbf{1.1(0.0)E-1} & 9.2(0.0)E-2 & 9.6(0.0)E-1 & 3.7(0.0)E-1 & 1.1(0.0)E-1 \\
    \rowcolor{gray!20} nes1980-nes & 10    & \textbf{3.1(0.1)E-1} & \textbf{1.1(0.0)E-1} & 1.2(0.0)E-1 & 9.6(0.0)E-1 & 4.9(0.1)E-1 & 1.1(0.0)E-1 \\
    nes1984-nes & 10    & \textbf{2.4(0.0)E-1} & \textbf{1.2(0.0)E-1} & 9.5(0.1)E-2 & 9.5(0.0)E-1 & 3.7(0.0)E-1 & 1.1(0.0)E-1 \\
    nes1988-nes & 10    & \textbf{2.5(0.1)E-1} & \textbf{1.1(0.0)E-1} & 1.0(0.0)E-1 & 9.7(0.0)E-1 & 3.9(0.0)E-1 & 1.0(0.0)E-1 \\
    nes1992-nes & 10    & \textbf{2.3(0.0)E-1} & \textbf{1.1(0.0)E-1} & 8.4(0.0)E-2 & 9.6(0.0)E-1 & 3.3(0.1)E-1 & 1.0(0.0)E-1 \\
    \rowcolor{gray!20} nes1996-nes & 10    & \textbf{2.6(0.0)E-1} & \textbf{1.1(0.0)E-1} & 9.6(0.1)E-2 & 9.9(0.0)E-1 & 3.9(0.0)E-1 & 1.1(0.0)E-1 \\
    nes2000-nes & 10    & \textbf{4.2(0.1)E-1} & \textbf{1.2(0.0)E-1} & 1.6(0.0)E-1 & 9.9(0.0)E-1 & 6.5(0.1)E-1 & 1.1(0.0)E-1 \\
    gp\_pois\_regr-gp\_pois\_regr & 13    & 2.6(1.4)E-2 & 1.5(0.2)E0 & \textbf{1.9(0.0)E-1} & \textbf{1.2(0.0)E0} & 4.9(0.2)E-1 & 1.1(0.0)E-1 \\
    diamonds-diamonds & 26    & 4.1(2.1)E-4 & 2.0(0.1)E0 & \textbf{7.6(0.5)E-2} & \textbf{1.5(0.0)E0} & 3.4(0.4)E-1 & 3.1(2.0)E-1 \\
    mcycle\_gp-accel\_gp & 66    & 0 & 1.9(0.0)E0 & \textbf{3.2(0.2)E-1} & \textbf{1.3(0.0)E0} & 0 & 1.8(0.0)E0 \\
    \bottomrule
    \end{tabular}}%
  \caption{Benchmarking using \texttt{PosteriorDB}. 
  Here, we compared a gradient-free version of \ac{rlmh} to the gradient-free \ac{arwmh}, and also the gradient-based \ac{amala}. 
  Performance was measured using the \acf{esjd} and the \acf{mmd} relative to the gold-standard, and $d = \textrm{dim}(\mathcal{X})$. 
  Results are based on an average of 10 replicates, with standard errors (in parentheses) reported.
  The best performing gradient-free method is highlighted in \textbf{bold}. 
  Shaded rows indicate situations where \ac{rlmh} out-performed \ac{amala} for either \ac{esjd} or \ac{mmd}.
  }
  \label{tab: posteriordb}%
\end{table}%

\texttt{PosteriorDB} is a community benchmark for performance assessment in Bayesian computation, consisting of a collection of posteriors to be numerically approximated \citep{magnusson2022posterior}.
Here we use \texttt{PosteriorDB} to compare gradient-free \ac{arwmh} (using default settings detailed in \Cref{app: adaptive mcmc}) against \ac{rlmh} (default settings in \Cref{subsec: training detail}).
A plethora of other algorithms exist, but \ac{arwmh} represents arguably the most widely-used gradient-free algorithm for adaptive \ac{mcmc}.
As an additional point of reference, we also present results for an adaptive version of \acuse{amala} \ac{mala} (AMALA; \Cref{subsec: mala}), for which gradient information on $p(\cdot)$ is required.
For higher-dimensional problems gradient information is usually essential.
For assessment purposes, \ac{rlmh} was run for $n = 5 \times 10^4$ iterations, while \ac{arwmh} and \ac{amala} were each run for $n = 6 \times 10^4$ iterations; this equates computational cost when one accounts for the warm start of \ac{rlmh}.
At the end, all algorithms were then run for an additional $5 \times 10^3$  iterations with no adaptation permitted, and it was on these final samples that performance was assessed.
Two metrics are reported:  (i) \ac{esjd}, and (ii) the \ac{mmd} relative to a gold-standard provided as part of \texttt{PosteriorDB}.
Both performance metrics are precisely defined in \Cref{app: performance measures}.
Full results are provided in \Cref{subsec: full results}.
These results show the gradient-free version of \ac{rlmh} out-performed the natural gradient-free comparator, \ac{arwmh}, on 86\% of tasks in terms of \ac{esjd}, and 93\% of tasks in terms of \ac{mmd}.
Remarkably, the gradient-free version of \ac{rlmh} also out-performed \ac{amala} on the majority of low-dimensional tasks, while \ac{amala} demonstrated predictably superior performance on higher-dimensional tasks where gradient information is well-known to be essential.
\ac{rlmh} and \ac{amala} both failed on the most challenging $66$-dimensional task, with \ac{rlmh} converging to a policy for which all subsequent samples were rejected.

\section{Discussion}

This paper provided, for the first time, a correct framework that enables modern techniques from \ac{rl} to be brought to bear on adaptive \ac{mcmc}.
Though the framework is general, for the purposes of end-to-end theoretical analysis we focused on a gradient-free sampling algorithm whose state-dependent proposal mean function is actively learned.
Even in this context, an astonishing level of performance was observed on the \texttt{PosteriorDB} benchmark when we consider that we did not exploit gradient information on the target, and that an off-the-shelf implementation of \ac{ddpg} was used.
Of course, gradient-free sampling algorithms are limited to tasks that are low-dimensional, and a natural next step is to investigate the extent to which performance can be improved by exploiting gradient information in the proposal.

More broadly, our contribution comes at a time of increasing interest in exploiting \ac{rl} for adaptive Monte Carlo methods, with recent work addressing adaptive importance sampling \citep{el2021policy} and the adaptive design of control variates \citep{bras2023policy}.
It would be interesting to see whether the approach we have set out can be extended to the design of other related Monte Carlo algorithms, such as multiple-try Metropolis \citep{liu2000multiple}, and delayed acceptance \ac{mcmc} \citep{christen2005markov}.

\paragraph{Acknowledgements}
CW was supported by the China Scholarship Council under Grand Number 202208890004. HK and CJO were supported by EP/W019590/1.

\FloatBarrier


\newpage
\appendix

\section*{Appendices}

\Cref{app: proofs} contains the proofs for all theoretical results stated in the main text.
\Cref{app: implementation} contains full details of our implementation, so that the empirical results we report can be reproduced.
Full empirical results are contained in \Cref{app: full empirical}.

\section{Proofs}
\label{app: proofs}

\Cref{app: ergodic proof} contains the proof of \Cref{thm: explicit}.
\Cref{app: general proof} contains the prof of \Cref{lem: ssge for phi mh}.
\Cref{app: proof of explicit theorem} contains the proof of \Cref{thm: explicit}.
Auxiliary lemmas used for these proofs are contained in \Cref{app: auxiliary}.

\subsection{Proof of \Cref{lem: ergodic}}
\label{app: ergodic proof}

\begin{proof}[Proof of \Cref{lem: ergodic}]
    Since $\phi$-MH is a Metropolis--Hastings chain it is automatically $p$-invariant.
    The proposal distribution of $\phi$-MH has a density $q_{\phi(x)}(\cdot | x)$.
    Under our assumptions, $(x,y) \mapsto q_{\phi(x)}(y|x)$ is positive and continuous over $x,y \in \mathcal{X}$.
    It follows that $\phi$-MH is both (a) aperiodic, and (b) $p$-irreducible; in addition, from Corollary 2 of \citet{tierney1994markov}, $\phi$-MH is Harris recurrent.
    The conclusion then follows from Theorem 1 of \citet{tierney1994markov}.
\end{proof}

\subsection{Proof of \Cref{lem: ssge for phi mh}}
\label{app: general proof}

The minorisation and drift conditions in \Cref{def: ssge} must be established.
These are the content, respectively, of \Cref{lem: minorisation} and \Cref{lem: drift condition}.
In the sequel we let $\lambda_{\mathrm{Leb}}(C)$ denote the Lebesgue measure of a set $C \subset \mathbb{R}^d$, and let $R_\theta(x) = \mathbb{R}^d \setminus A_\theta(x)$ denote the region where proposals may be rejected.
The following argument builds on relatively standard arguments used to establish ergodicity of Metropolis--Hastings, such as Lemma 1.2 of \citet{mengersen1996rates}, but includes an additional dependence on the parameter $\theta \in \Theta$:

\begin{lemma}[Simultaneous minorisation condition for Metropolis--Hastings]
\label{lem: minorisation}
    Let $\mathcal{X} = \mathbb{R}^d$.
    Let $\Theta$ be a topological space and let $\Theta_0 \subset \Theta$ be compact.
    Consider a family of $p$-invariant Metropolis--Hastings transition kernels $P_\theta$, and corresponding proposals $\{q_\theta(\cdot|x) : x \in \mathbb{R}^d \}$, indexed by $\theta \in \Theta$.
    Let $x \mapsto p(x)$ be positive and continuous, let $(\theta,x,y) \mapsto q_\theta(y|x)$ be positive and continuous over $\theta \in \Theta$, $x,y \in \mathbb{R}^d$, and let $C \subset \mathbb{R}^d$ be compact with $\lambda_{\mathrm{Leb}}(C) > 0$.
    Then there exists $\delta > 0$ and a probability measure $\nu$ on $C$, such that $P_\theta(x,\cdot) \geq \delta \nu(\cdot)$ holds simultaneously for all $\theta \in \Theta_0$ and all $x \in C$.
\end{lemma}
\begin{proof}
    From positivity, continuity and compactness $p_{\sup} := \sup_{x \in C} p(x) \in (0,\infty)$ and $q_{\inf} := \inf_{x,y \in C, \theta \in \Theta_0} q_\theta(y|x) \in (0,\infty)$.
    Let $p(C) := \int_C p(x) \mathrm{d}x$, which is positive since $p(\cdot)$ is bounded away from 0 on $C$ and $\lambda_{\mathrm{Leb}}(C) > 0$.
    For the result we will take $\delta = (q_{\inf} / p_{\sup}) p(C)$ and $\nu(\cdot) = p(\cdot) / p(C)$.
    Then for all $\theta \in \Theta_0$, $x \in C$, and $S \subseteq C$,
    \begin{align*}
        P_\theta(x,S) & = \underbrace{ \int_S q_\theta(y|x) \alpha_{\theta}(x,y) \mathrm{d}y }_{\text{probability of moving from } x \text{ to a different state in } S} \; + \;  \mathrm{1}_{S}(x) \underbrace{ \int_{\mathcal{X}} q_\theta(y|x) [ 1 - \alpha_\theta(x,y) ] \mathrm{d}y }_{\text{probability of rejecting a proposal}} \\
        & \geq \int_S q_\theta(y|x) \alpha_{\theta}(x,y) \mathrm{d}y \\
        & = \int_{S \cap A_\theta(x)} q_\theta(y|x) \alpha_{\theta}(x,y) \mathrm{d}y + \int_{S \cap R_\theta(x)} q_\theta(y|x) \alpha_{\theta}(x,y) \mathrm{d}y \\
        & = \int_{S \cap A_\theta(x)} q_\theta(y|x) \mathrm{d}y + \int_{S \cap R_\theta(x)} \frac{p(y)}{p(x)}  q_\theta(x|y) \mathrm{d}y \\
        & \geq \int_{S \cap A_\theta(x)}  \frac{p(y)}{p_{\sup}} q_{\inf} \mathrm{d}y + \int_{S \cap R_\theta(x)}  \frac{p(y)}{p_{\sup}} q_{\inf} \mathrm{d}y
        = \int_S \frac{p(y)}{p_{\sup}} q_{\inf} \mathrm{d}y
        = \delta \nu(S) ,
    \end{align*}
    as required.
\end{proof}
\noindent Note that the conclusion of \Cref{lem: minorisation} is slightly stronger than what we are minimally required to establish for minorisation in \Cref{def: ssge}, since it provides a probability measure $\nu$ that applies simultaneously for all $\theta$ in a compact set $\Theta$.

Our focus now turns to the drift condition.
For $x \in \mathbb{R}^d$ and $R \geq 0$, let $B_R(x) := \{y \in \mathbb{R}^d : \|y - x\| \leq R\}$ denote the $x$-centred, radius $R$ ball.
The following is based on an argument used in the proof of Theorem 4.1 of \citet{jarner2000geometric}, but generalised beyond the case of a random walk proposal.
For the proof we will use a technical lemma on the tail properties of sub-exponential distributions, provided as \Cref{lem: super exp tails} in \Cref{app: auxiliary}.

\begin{lemma}[Simultaneous drift condition for Metropolis--Hastings]
\label{lem: drift condition}
    In the setting of \Cref{lem: ssge for phi mh}, there exists $V : \mathbb{R}^d \rightarrow [1,\infty)$, $\lambda < 1$, $b < \infty$, and $R > 0$ such that, letting $S = B_R(0)$, we have $\sup_{x \in S} V(x) < \infty$ and $(P_\theta V)(x) \leq \lambda V(x) + b \mathrm{1}_S(x)$ for all $x \in \mathbb{R}^d$.
\end{lemma}
\begin{proof}
    First note that our quasi-symmetry assumption implies that
\begin{align}
\left\{y : \frac{p(y)}{p(x)} \geq \rho \right\} & \subseteq A_\theta(x) \subseteq \left\{y : \frac{p(y)}{p(x)} \geq \frac{1}{\rho} \right\} \label{eq: quasi sym A} \\
\left\{y : \frac{p(y)}{p(x)} < \frac{1}{\rho} \right\} & \subseteq R_\theta(x) \subseteq \left\{y : \frac{p(y)}{p(x)} < \rho \right\} \label{eq: quasi sym R}
\end{align}
and also that $x \in A_\theta(x)$ will always hold.
For the proof we will take $V(x) = c p(x)^{-1/2}$ for $c$ such that $V \geq 1$, which we can do since $p$ is continuous, positive, and vanishing in the tail.
It follows from continuity and compactness that $\sup_{x \in S} V(x) < \infty$ is satisfied.
It then suffices to show that
\begin{align}
\limsup_{\|x\| \rightarrow \infty} \sup_{\theta \in \Theta} \frac{(P_\theta V)(x)}{V(x)} < 1 \label{eq: ergodic tail} \\
\sup_{x \in \mathbb{R}^d} \sup_{\theta \in \Theta} \frac{(P_\theta V)(x)}{V(x)} < \infty . \label{eq: ergodic bound}
\end{align}
Indeed, if \eqref{eq: ergodic tail} holds then there exists $R > 0$, $\lambda < 1$ such that $P_\theta V(x) \leq \lambda V(x)$ for all $x \notin S$ and all $\theta \in \Theta$.
Further, if \eqref{eq: ergodic bound} holds then, since $V$ is continuous, we can set
$$
b = \sup_{x \in S} \sup_{\theta \in \Theta} (P_\theta V)(x) \leq \sup_{x \in S} V(x) \sup_{x' \in \mathbb{R}^d} \sup_{\theta \in \Theta} \frac{(P_\theta V)(x')}{V(x')} < \infty
$$
since $V$ is bounded on the compact set $S$.
Then the drift condition will have been established.

\smallskip

\noindent \textit{Establishing \eqref{eq: ergodic bound}:}
Decompose and then bound the integral as
\begin{align*}
    (P_\theta V)(x) & = \underbrace{ \int_{A_\theta(x)} q_\theta(y|x) V(y) \mathrm{d}y }_{\text{always accept}} \\
    & \qquad + \underbrace{ \int_{R_\theta(x)} q_\theta(y | x) \left[ \frac{q_\theta(x|y) p(y)}{q_\theta(y|x) p(x)} V(y) + \left( 1 - \frac{q_\theta(x|y) p(y)}{q_\theta(y|x) p(x)} \right) V(x) \right] \mathrm{d}y }_{\text{possibly reject}} \\
    & \leq \int_{A_\theta(x)} q_\theta(y|x) V(y) \mathrm{d}y + \int_{R_\theta(x)} \frac{q_\theta(x|y) p(y)}{p(x)} V(y) + q_\theta(y|x) V(x)  \mathrm{d}y
\end{align*}
so
\begin{align*}
    \frac{(P_\theta V)(x)}{V(x)} & \leq \int_{A_\theta(x)} q_\theta(y|x) \frac{V(y)}{V(x)} \mathrm{d}y + \int_{R_\theta(x)} \frac{q_\theta(x|y) p(y)}{p(x)}  \frac{V(y)}{V(x)} + q_\theta(y|x) \mathrm{d}y .
\end{align*}
(The first of these two inequalities is in fact strict, but we do not need a strict inequality for our argument.)
Then, plugging in our choice of $V$, we have
\begin{align}
    \frac{(P_\theta V)(x)}{V(x)} & \leq \int_{A_\theta(x)} q_\theta(y|x) \frac{p(x)^{1/2}}{p(y)^{1/2}} \mathrm{d}y + \int_{R_\theta(x)} q_\theta(x|y) \frac{p(y)^{1/2}}{p(x)^{1/2}}  + q_\theta(y|x) \mathrm{d}y . \label{eq: explicit ratio form}
\end{align}
From \eqref{eq: quasi sym A} and \eqref{eq: quasi sym R}, together with $q_\theta(x|y) \leq \rho q_\theta(y|x)$, we see that
\begin{align*}
    \frac{(P_\theta V)(x)}{V(x)} & \leq \int_{A_\theta(x)} q_\theta(y|x) \rho^{1/2} \mathrm{d}y + \int_{R_\theta(x)} \rho q_\theta(y|x) \rho^{1/2} + q_\theta(y|x) \mathrm{d}y \\
    & = \rho^{1/2} Q_\theta(A_\theta(x) | x ) + (\rho^{3/2} + 1) Q_\theta(R_\theta(x) | x ) 
    = \rho^{1/2} + \rho^{3/2} + 1
    < \infty
\end{align*}
where the final bound is $x$- and $\theta$-independent, so that \eqref{eq: ergodic bound} is established.

\smallskip

\noindent \textit{Establishing \eqref{eq: ergodic tail}:}
Fix $\epsilon > 0$.
By the regular acceptance boundary assumption, there exists $\delta > 0$ small enough and $R_1 > 0$ large enough that $Q_\theta(\partial A_\theta^\delta(x) | x ) < \epsilon$ for all $\theta$ and all $\|x\| \geq R_1$.
From \Cref{lem: super exp tails}, there exists $R_2 > 0$ large enough that, for all $\|x\| \geq R_2$, $y \in A_\theta(x) \cap ( \partial A_\theta^\delta(x)^\mathtt{c} )$ implies that $p(x) / p(y) \leq \epsilon$ and $y \in R_\theta(x) \cap \partial A_\theta^\delta(x)^c$ implies that $p(y) / p(x) \leq \epsilon$.
So, for $\|x\| \geq \max\{R_1,R_2\}$, and using \eqref{eq: quasi sym A},
\begin{align*}
    \int_{A_\theta(x)} q_\theta(y|x) \frac{p(x)^{1/2}}{p(y)^{1/2}} \mathrm{d}y & = \int_{A_\theta(x) \cap \partial A_\theta^\delta(x) } q_\theta(y|x) \frac{p(x)^{1/2}}{p(y)^{1/2}} \mathrm{d}y \\
    & \qquad + \int_{A_\theta(x) \cap ( \partial A_\theta^\delta(x)^{\mathtt{c}} ) } q_\theta(y|x) \frac{p(x)^{1/2}}{p(y)^{1/2}} \mathrm{d}y \\
    & \leq \rho^{1/2} Q_\theta(\partial A_\theta^\delta(x) | x ) + \epsilon^{1/2} Q_\theta(\mathbb{R}^d | x) 
    \leq \rho^{1/2} \epsilon + \epsilon^{1/2} .
\end{align*}
Since $\epsilon > 0$ was arbitrary, we have shown that
\begin{align}
\limsup_{\|x\| \rightarrow \infty} \sup_{\theta \in \Theta} \int_{A_\theta(x)} q_\theta(y|x) \frac{p(x)^{1/2}}{p(y)^{1/2}} \mathrm{d}y = 0 . \label{eq: first integral vanishes}
\end{align}
A similar argument shows that, for $\|x\| \geq \max\{R_1,R_2\}$, and using \eqref{eq: quasi sym R},
\begin{align*}
\int_{R_\theta(x)} q_\theta(x|y) \frac{p(y)^{1/2}}{p(x)^{1/2}}  \mathrm{d}y 
& \leq \int_{R_\theta(x) \cap \partial A_\theta^\delta(x)} q_\theta(x|y) \frac{p(y)^{1/2}}{p(x)^{1/2}} \mathrm{d}y \\
& \qquad + \int_{R_\theta(x) \cap ( \partial A_\theta^\delta(x)^{\mathtt{c}} ) } q_\theta(x|y) \frac{p(y)^{1/2}}{p(x)^{1/2}}  \mathrm{d}y \\
& \leq \rho Q_\theta(\partial A_\theta^\delta(x) | x) \rho^{1/2} + \rho Q_\theta(\mathbb{R}^d | x) \epsilon^{1/2} 
\leq \rho^{3/2} \epsilon + \rho \epsilon^{1/2}
\end{align*}
and since $\epsilon > 0$ was arbitrary, we have shown that
\begin{align}
\limsup_{\|x\| \rightarrow \infty} \sup_{\theta \in \Theta} \int_{R_\theta(x)} q_\theta(x|y) \frac{p(y)^{1/2}}{p(x)^{1/2}} \mathrm{d}y = 0 . \label{eq: second integral vanishes}
\end{align}
Thus, substituting \eqref{eq: first integral vanishes} and \eqref{eq: second integral vanishes} into \eqref{eq: explicit ratio form} yields 
\begin{align*}
    \limsup_{\|x\| \rightarrow \infty} \sup_{\theta \in \Theta} \frac{P_\theta V(x)}{V(x)} & = \limsup_{\|x\| \rightarrow \infty} \sup_{\theta \in \Theta} \int_{R_\theta(x)} q_\theta(y|x) \mathrm{d}y \\
    & = 1 - \liminf_{\|x\| \rightarrow \infty} \inf_{\theta \in \Theta} Q_\theta(A_\theta(x) | x) < 1,
\end{align*}
where we have used the minimum performance assumption to obtain the inequality. 
The claim \eqref{eq: ergodic tail} has been established. 
\end{proof}

\begin{proof}[Proof of \Cref{lem: ssge for phi mh}]
From \Cref{def: ssge} we must show that minorisation and simultaneous drift conditions are satisfied.
The simultaneous drift condition is established in \Cref{lem: drift condition} with $S = B_R(0)$ for some $R > 0$.
Since $S$ is compact with $\lambda_{\text{Leb}}(S) > 0$, the minorisation condition follows from \Cref{lem: minorisation}.
\end{proof}

\subsection{Proof of \Cref{thm: explicit}}
\label{app: proof of explicit theorem}

The proof of \Cref{thm: explicit} exploits the framework for analysis of adaptive \ac{mcmc} advocated in \citealp{roberts2007coupling}.
Auxiliary lemmas used in this proof are deferred to \Cref{app: auxiliary}.
For a (possibly signed) measure $\nu$ on $\mathcal{X}$, denote the total variation norm $\|\nu\|_{\mathrm{TV}} = \sup_{S \subset \mathcal{X}} |\nu(S)|$.
Then we aim to make use of the following well-known result:

\begin{theorem}[Theorem 3 of \citealp{roberts2007coupling}]
\label{thm: roberts}
    Let $\Theta_0$ be a set and consider an adaptive \ac{mcmc} algorithm $\theta_n \equiv \theta_n(x_0,\dots,x_n) \in \Theta_0$ with Markov transition kernels $\{P_\theta\}_{\theta \in \Theta_0}$ initialised at a fixed $x_0 \in \mathcal{X}$ and $\theta_0 \in \Theta_0$.
    Suppose $\{P_\theta\}_{\theta \in \Theta_0}$ is \ac{ssage} and that the \emph{diminishing adaptation} condition, meaning
    $$
    \sup_{x \in \mathcal{X}} \| P_{\theta_{n+1}}(x,\cdot) - P_{\theta_n}(x,\cdot) \|_{\mathrm{TV}} \rightarrow 0
    $$
    in probability as $n \rightarrow \infty$, is satisfied.
    Then $\| \mathrm{Law}(x_n) - p \|_{\mathrm{TV}} \rightarrow 0$.
\end{theorem}

The easier of these two conditions to establish is diminishing adaptation, which is satisfied due to our control of the learning rate and clipping of the gradient, as demonstrated in \Cref{lem: diminishing}.
A useful fact that we will use in the proof is that the norms $\|\cdot\|$ and $\|\cdot\|_{1,\Sigma}$ are equivalent with
\begin{align}
\lambda_{\max}^{-1/2}(\Sigma) \|x\| \leq \|x\|_{1,\Sigma} \leq \sqrt{d} \lambda_{\min}^{-1/2}(\Sigma) \|x\| \label{eq: norm equiv}
\end{align}
for all $x \in \mathbb{R}^d$, where $\lambda_{\min}(\Sigma)$ and $\lambda_{\max}(\Sigma)$ denote, respectively, the minimum and maximum eigenvalues of the matrix $\Sigma \in \mathrm{S}_d^+$.

\begin{lemma}[Diminishing adaptation for \ac{rlmh}]
\label{lem: diminishing}
    The gradient clipping with threshold $\tau > 0$ and summable learning rate $(\alpha_n)_{n \geq 0} \subset [0,\infty)$, appearing in \Cref{alg: rlmh}, ensure that
    \begin{align}
        (\theta_n)_{n \geq 0} \subset \Theta_0 := B_T(\theta_0), \qquad T := \tau \sum_{n = 0}^\infty \alpha_n < \infty . \label{eq: compact theta set}
    \end{align} 
    From local boundedness it follows that
    \begin{align}
        B := \sup_{\theta \in \Theta_0} \mathrm{Lip}(\phi_\theta) < \infty ,  \label{eq: uniform lips}
    \end{align}
    since $\Theta_0$ is compact.   
    In particular, in the setting of \Cref{thm: explicit}, diminishing adaptation is satisfied.
\end{lemma}
\begin{proof}
    The set-up in \Cref{alg: rlmh} implies that
    \begin{align*}
    \|\theta_n - \theta_0\| \leq \sum_{i=0}^{n-1} \|\theta_{i+1} - \theta_i\| \leq \sum_{i=0}^{n-1} \alpha_i \tau \leq \tau \sum_{i=0}^\infty \alpha_i < \infty 
    \end{align*}
    where the final bound is $n$-independent and finite, since the summability of the learning rate $(\alpha_n)_{n \geq 0} \subset [0,\infty)$ was assumed.
    From this, \eqref{eq: compact theta set} is immediately established.

    The main idea of this proof is to exploit the triangle inequality and the definition of the total variation norm, as follows:
    \begin{align*}
        \| P_\theta(x,\cdot) - P_\vartheta(x,\cdot) \|_{\mathrm{TV}} & \leq \left\| \int q_\theta(y|x) [ \delta_y(\cdot) \alpha_\theta(x,y) + \delta_x(\cdot) (1 - \alpha_\theta(x,y)) ] \; \mathrm{d}y  \right. \\ & \left. \hspace{30pt} - \int q_\vartheta(y|x) [ \delta_y(\cdot)\alpha_\vartheta(x,y) + \delta_x(\cdot) (1 - \alpha_\vartheta(x,y)) ] \; \mathrm{d}y  \right\|_{\mathrm{TV}} \\
        & \leq \left\| \int {\delta}_y(\cdot)[ q_\theta(y|x) \alpha_\theta(x,y) - q_\vartheta(y|x)  \alpha_\vartheta(x,y) ] \; \mathrm{d}y  \right\|_{\mathrm{TV}} \\
        & \hspace{30pt} + \left\| \delta_x(\cdot) \int [ q_\theta(y|x) \alpha_\theta(x,y) - q_\vartheta(y|x)  \alpha_\vartheta(x,y) ] \; \mathrm{d}y  \right\|_{\mathrm{TV}} \\
        & \leq \int | q_\theta(y|x) \alpha_\theta(x,y) - q_\vartheta(y|x)  \alpha_\vartheta(x,y) | \; \mathrm{d}y ,
        \end{align*}
        where $\delta_x$ denotes the probability distribution that puts all mass at $x \in \mathbb{R}^d$.
    From the triangle inequality again,
    \begin{align}
        \| P_\theta(x,\cdot) - P_\vartheta(x,\cdot) \|_{\mathrm{TV}} & \leq \int \left| q_\theta(y|x) - q_\vartheta b(y|x) \right| \alpha_{\theta}(x,y) + |\alpha_{\theta}(x,y) - \alpha_{\vartheta}(x,y)| q_\vartheta(y|x) \; \mathrm{d}y \nonumber \\
        & \hspace{-30pt}  \leq  \int \left| q_\theta(y|x) - q_\vartheta(y|x) \right| \; \mathrm{d}y + \int \left| \alpha_{\theta}(x,y) - \alpha_{\vartheta}(x,y) \right| q_\vartheta(y|x) \; \mathrm{d}y , \label{eq: main bound in dim arg}
    \end{align}
    and we seek to bound the two integrals appearing in \eqref{eq: main bound in dim arg}.
   
    Fix $x \in \mathbb{R}^d$.
    Then the map $\vartheta \mapsto \phi_\vartheta(x)$ is locally Lipschitz, and since $\Theta_0$ is connected and compact it follows that $\vartheta \mapsto \phi_\vartheta(x)$ is Lipschitz on $\Theta_0$, with Lipschitz constant 
    $$
    L_x := \sup_{\theta \in \Theta_0} \mathrm{LocLip}_\theta(\vartheta \mapsto \phi_\vartheta(x)) .
    $$    
    Further, this Lipschitz constant can be uniformly bounded over $x \in \mathbb{R}^d$:
    \begin{align*}
    L := \sup_{x \in \mathbb{R}^d} L_x & = \sup_{x \in \mathbb{R}^d} \sup_{\theta \in \Theta_0} \mathrm{LocLip}_\theta(\vartheta \mapsto \phi_\vartheta(x)) \\
    & = \sup_{\theta \in \Theta_0} \sup_{x \in \mathbb{R}^d} \mathrm{LocLip}_\theta(\vartheta \mapsto \phi_\vartheta(x))
    < \infty ,
    \end{align*}
    where finiteness follows since we assumed local boundedness of $\theta \mapsto \sup_{x \in \mathbb{R}^d} \mathrm{LocLip}_\theta(\vartheta \mapsto \phi_\vartheta(x))$ and $\Theta_0$ is compact.
  
    Fix $\theta, \vartheta \in \Theta_0$.
    Now,
    \begin{align*}
        |q_\theta(y|x) - q_\vartheta(y|x)| & = q_\vartheta(y|x) \left| 1 - \frac{q_\theta(y|x)}{q_\vartheta(y|x)} \right| \\
        & = q_\vartheta(y|x) \left| 1 - \exp\left( -\|y - \phi_\theta(x)\|_{1,\Sigma} + \|y - \phi_\vartheta(x)\|_{1,\Sigma} \right) \right| .
    \end{align*}
    From the reverse triangle inequality, the aforementioned Lipschitz property, and \eqref{eq: norm equiv},
    \begin{align}
        \left| - \|y - \phi_\theta(x)\|_{1,\Sigma} + \|y - \phi_\vartheta(x)\|_{1,\Sigma}  \right| & \leq \| \phi_\theta(x) - \phi_\vartheta(x) \|_{1,\Sigma} 
        \leq \sqrt{d} \lambda_{\min}^{1/2}(\Sigma) \;  L \; \| \theta - \vartheta \| . \label{eq: local lips}
    \end{align}
    Further, our assumptions imply that the right hand side of \eqref{eq: local lips} is bounded since $\theta,\vartheta \in \Theta_0$ and $\Theta_0$ is compact.
    Since the exponential function is Lipschitz on any compact set, there exists an $(x,\theta,\vartheta)$-independent constant $C > 0$ such that
    \begin{align*}
        \left| 1 -  \exp\left( -\|y - \phi_\theta(x)\|_{1,\Sigma} + \|y - \phi_\vartheta(x)\|_{1,\Sigma} \right) \right| \leq C  \left| - \|y - \phi_\theta(x)\|_{1,\Sigma} + \|y - \phi_\vartheta(x)\|_{1,\Sigma}  \right|  
    \end{align*}
    and hence
    \begin{align}
        \int |q_\theta(y|x) - q_\vartheta(y|x)| \; \mathrm{d}y & \leq C \; \sqrt{d} \lambda_{\min}^{1/2}(\Sigma) \; L \; \| \theta - \vartheta \|  \label{eq: bound integral 1}
    \end{align}
    holds for all $\theta,\vartheta \in \Theta_0$ and all $x \in \mathbb{R}^d$.

    Next, from \Cref{lem: quasi-symmetric} the proposal $\{q_\theta(\cdot | x) : x \in \mathbb{R}^d\}$ is quasi-symmetric on $\Theta_0$ with constant $\rho$ defined as in \eqref{eq: quasi sym const}.
    From \eqref{eq: quasi sym A}, if $p(y) / p(x) \geq \rho$ then $y \in A_\theta(x) \cap A_\vartheta(x)$ and $\alpha_{\theta}(x,y) = \alpha_{\vartheta}(x,y) = 1$, so that $| \alpha_{\theta}(x,y) - \alpha_{\vartheta}(x,y) | = 0$.
    Otherwise, we have $p(y) / p(x) < \rho$ and from quasi-symmetry again,
    \begin{align*}
        \left| \alpha_{\theta}(x,y) - \alpha_{\vartheta}(x,y) \right| & \leq \frac{p(y)}{p(x)} \left| \frac{q_\theta(x|y)}{q_\theta(y|x)} - \frac{q_\vartheta(x|y)}{q_\vartheta(y|x)} \right| \\
        & = \frac{p(y)}{p(x)} \frac{q_\theta(x|y)}{q_\theta(y|x)} \left| 1 - \frac{q_\theta(y|x)}{q_\theta(x|y)} \frac{q_\vartheta(x|y)}{q_\vartheta(y|x)} \right| \\
        & \leq \rho^2 \left| 1 - \exp\left( \begin{array}{l} - \|y - \phi_\theta(x)\|_{1,\Sigma} + \|x-\phi_\theta(y)\|_{1,\Sigma}  \\ \qquad  -\|x-\phi_\vartheta(y)\|_{1,\Sigma} + \|y - \phi_\vartheta(x)\|_{1,\Sigma} \end{array} \right) \right| .
    \end{align*}
    From \eqref{eq: local lips} and an analogous argument to before involving the fact that the exponential function is Lipschitz on a compact set, we obtain that
    \begin{align}
        \int \left| \alpha_{\theta}(x,y) - \alpha_{\vartheta}(x,y) \right| q_\vartheta(y|x) \mathrm{d}y & \leq 2 C \rho^2 \; \sqrt{d} \lambda_{\min}^{1/2}(\Sigma) \; L \; \| \theta - \vartheta \|  \label{eq: bound integral 2}
    \end{align}
    for all $\theta , \vartheta \in \Theta_0$ and $x \in \mathbb{R}^d$.

    Substituting \eqref{eq: bound integral 1} and \eqref{eq: bound integral 2} into \eqref{eq: main bound in dim arg},
    \begin{align*}
        \| P_\theta(x,\cdot) - P_\vartheta(x,\cdot) \|_{\mathrm{TV}} & \leq C (1 + 2 \rho^2) \; \sqrt{d} \lambda_{\min}^{1/2}(\Sigma) \; L \; \| \theta - \vartheta \| 
    \end{align*}
    for all $\theta , \vartheta \in \Theta_0$ and $x \in \mathbb{R}^d$.
    It follows (a.s.) that 
    \begin{align*}
        \sup_{x \in \mathbb{R}^d} \| P_{\theta_n}(x,\cdot) - P_{\theta_{n+1}}(x,\cdot) \|_{\mathrm{TV}}  & \leq C (1 + 2 \rho^2) \; \sqrt{d} \lambda_{\min}^{1/2}(\Sigma) \; L \; \|\theta_n - \theta_{n+1}\| \rightarrow 0
    \end{align*}
    since $(\theta_n)_{n \geq 0}$ is (a.s.) convergent.  Since almost sure convergence implies convergence in probability, diminishing adaptation is established.  
\end{proof}
The main technical effort required to prove \Cref{thm: explicit} occurs in establishing conditions under which $\phi$-MH is \ac{ssage}.
This result is the content of \Cref{lem: explicit SSGE}.
For the proof we will use technical lemmas on the tail properties of sub-exponential distributions, provided as \Cref{lem: super exp tails,lem: quasi-symmetric,lem: geometry of Ax}, and a technical lemma on the interior cone condition, provided as \Cref{lem: cone}, all of which can be found in \Cref{app: auxiliary}.

For a subset $C \subset \mathbb{R}^d$, let $C \cong \mathbb{S}^{d-1}$ denote that $C = \{r(\xi) \xi : \xi \in \mathbb{S}^{d-1}\}$ for some continuous function $r : \mathbb{S}^{d-1} \rightarrow (0,\infty)$, meaning that $C$ is a hyper-surface
that can be parametrised using the $(d-1)$-dimensional sphere $\mathbb{S}^{d-1}$.
For $\epsilon > 0$, let $C_\epsilon := \{x \in \mathbb{R}^d : p(x) = \epsilon\}$ be the $\epsilon$-level set of $p(\cdot)$ and, for $\delta \geq 0$, let $C_\epsilon^\delta := \{x + sn(x) : x \in C_\epsilon, |s| \leq \delta \}$.

\begin{lemma}[Flexible mean $\phi$-MH is SSAGE]
\label{lem: explicit SSGE}
    Let $\Theta_0 \subset \Theta = \mathbb{R}^p$ be the compact set from \eqref{lem: diminishing}.
    In the setting of \Cref{thm: explicit}, $\{P_\theta\}_{\theta \in \Theta_0}$ is \ac{ssage}.
\end{lemma}
\begin{proof}
    The conditions of \Cref{lem: ssge for phi mh} need to be checked.
    Let $M := \sup_{\theta \in \Theta_0} \sup_{x \in \mathbb{R}^d} \|x - \phi_\theta(x) \|_{1,\Sigma}$, which exists by the assumed local boundedness of $\theta \mapsto \sup_{x \in \mathbb{R}^d} \|x - \phi_\theta(x) \|$, compactness of $\Theta_0$, and the norm-equivalence in \eqref{eq: norm equiv}.
    The form of our flexible mean Laplace proposal \eqref{eq: Laplace proposal} ensures that $q_{\sup} := \sup_{\theta \in \Theta_0,x,y \in \mathbb{R}^d} q_\theta(y|x) < \infty$.

    \smallskip

    \noindent\textit{Quasi-symmetry:}
    From \Cref{lem: quasi-symmetric}, the proposal $\{q_\theta(\cdot|x) : x \in \mathbb{R}^d\}$ is quasi-symmetric on $\Theta_0$, with the constant $\rho$ defined as in \eqref{eq: quasi sym const}.

    \smallskip

    \noindent\textit{Regular acceptance boundary:}
    Given $\epsilon > 0$, we can pick $R_1 > 0$ such that for all $x \in \mathbb{R}^d$ and all $\theta \in \Theta_0$,
    $$
    Q_\theta(B_{R_1}(\phi_\theta(x))^{\mathtt{c}} | x) < \epsilon
    $$
    due to the form of our Laplace proposal in \eqref{eq: Laplace proposal}.
    From the definition of $M$, we can pick $R_2 \geq R_1$ such that $B_{R_1}(\phi_\theta(x)) \subseteq B_{R_2}(x)$ holds simultaneously for all $x \in \mathbb{R}^d$ and all $\theta \in \Theta_0$.
    Thus in particular 
    \begin{align}
    Q_\theta(B_{R_2}(x)^{\mathtt{c}} | x) < \epsilon \label{eq: ball bound}
    \end{align}
    holds simultaneously for all $x \in \mathbb{R}^d$ and all $\theta \in \Theta_0$.
    
    From \Cref{lem: super exp tails,lem: geometry of Ax}, the assumption that $p \in \mathcal{P}(\mathbb{R}^d)$ implies there is an $R_3 > R_2$ such that for all $\|x\| \geq R_3$ we have  $C_{\rho p(x)} , \partial A_\theta(x),  C_{\rho^{-1} p(x)}  \cong \mathbb{S}^{d-1}$.
    Further, from \Cref{lem: super exp tails} with $r = \frac{4}{\epsilon} \log \rho$, and the fact $\rho \geq 1$, we may assume that $R_3$ is large enough that
    \begin{align*}
    \frac{p(x + s n(x))}{p(x)} \geq \frac{1}{\rho^2} \qquad \implies \qquad s \leq \frac{\epsilon}{2} 
    \end{align*} 
    so that the radial distance between $C_{\rho p(x)}$ and $C_{\rho^{-1} p(x)}$ is uniformly less than $\epsilon/2$ for all $\|x\| \geq R_3$.
    From \eqref{eq: quasi sym A}, both $C_{p(x)}$ and $\partial A_\theta(x)$ is contained in the region bounded by $C_{\rho p(x)}$ and $C_{\rho^{-1} p(x)}$.    
    It follows that, if $\delta \in [0,\epsilon/2)$, then $\partial A_\theta^\delta(x) \subset C_{p(x)}^{\epsilon}$ uniformly in $\|x\| \geq R_3$ and $\theta \in \Theta_0$.

    From \Cref{lem: lebesgue}, for all $\|x\| \geq R_3$ ($> R_2$) we have the bound
    \begin{align}
        \lambda_{\text{Leb}}\left( C_{p(x)}^\epsilon \cap B_{R_2}(x) \right) & \leq \epsilon \left( \frac{\|x\| + R_2}{\|x\| - R_2} \right)^{d-1} \frac{\lambda_{\text{Leb}}(B_{3R_2}(x))}{R_2} \nonumber \\
        & \leq \epsilon \underbrace{ \left( \frac{R_3 + R_2}{R_3 - R_2} \right)^{d-1} \frac{\lambda_{\text{Leb}}(B_{3R_2}(x))}{R_2} }_{=: D} \label{eq: leb bound}
    \end{align}
    where the last inequality is obtained by maximising the $x$-dependent term over $\|x\| \geq R_3$.

    Putting these results together, we have the bound
    \begin{align*}
        Q_\theta(\partial A_\theta^\delta(x) | x) & = Q_\theta(\partial A_\theta^\delta(x) \cap (B_{R_2}(x)^{\mathtt{c}}) | x ) + Q_\theta(\partial A_\theta^\delta(x) \cap B_{R_2}(x) | x )  \\
        & \leq \epsilon + q_{\sup} \; D \; \epsilon
    \end{align*}
    for all $\|x\| \geq R_3$ and all $\theta \in \Theta_0$, where for the first term we have used \eqref{eq: ball bound} and for the second term we have used \eqref{eq: leb bound}.
    Since $\epsilon > 0$ was arbitrary and our bound is $\theta$-independent, the regular acceptance boundary condition is established.

    \smallskip

    \noindent\textit{Minimum performance level:}
    From \Cref{lem: super exp tails,lem: geometry of Ax}, the assumption that $p \in \mathcal{P}(\mathbb{R}^d)$ implies there is an $R > 0$ such that for all $\|x\| \geq R$ we have  $C_{\rho p(x)} , \partial A_\theta(x),  C_{\rho^{-1} p(x)}  \cong \mathbb{S}^{d-1}$. 
    From \eqref{eq: quasi sym A}, $\partial A_\theta(x)$ is contained in the region bounded by $C_{\rho p(x)}$ and $C_{\rho^{-1} p(x)}$.
    For $x \neq 0$, let $y_x$ denote the intersection of the line $\{s x : s \geq 0\}$ with the set $C_{\rho p(x)}$.
    From \Cref{lem: super exp tails} and the norm-equivalence in \eqref{eq: norm equiv}, we may consider $R$ sufficiently large that $\|y_x - x\|_{1,\Sigma} \leq 1$ for all $\|x\| \geq R$, and in particular this leads to the bound $\|y_x - \phi_\theta(x)\|_{1,\Sigma} \leq M + 1$.

    From \Cref{lem: cone}, since $p \in \mathcal{P}_0(\mathbb{R}^d)$ we may assume $R$ is sufficiently large that for all $\|x\| \geq R$ there exists a radius-1 cone $K_{1}(y_x)$ with $x$-independent Lebesgue measure $\zeta_1 > 0$, whose apex is $y_x$, contained in the interior of the compact region bounded by $C_{\rho p(x)}$ ($\subset A_\theta(x)$).
    In particular, $K_1(y_x)$ is contained in a $\|\cdot\|_{1,\Sigma}$-ball of radius $M + 2$ centred at $\phi_\theta(x)$.
    
    Thus for all $\|x\| \geq R$ and $\theta \in \Theta_0$, 
    \begin{align*}
    Q_\theta(A_\theta(x) | x) \geq Q_\theta( A_\theta(x) \cap K_1(y_x) | x ) & = Q_\theta( K_1(y_x) | x ) \\
    & \geq \lambda_{\mathrm{Leb}}( K_1(y_x) ) \times \inf_{\|y - \phi_\theta(x) \|_{1,\Sigma} \leq M + 2} q_\theta(y|x) \\
    & \geq \zeta_1 \times \frac{1}{Z} \exp( - (M+2) ) > 0
    \end{align*}
    where $Z > 0$ is the $(x,\theta)$-independent normalisation constant of the proposal \eqref{eq: Laplace proposal}.
    Thus the final condition of \Cref{lem: ssge for phi mh} is satisfied.
\end{proof}

At last we have established \Cref{thm: explicit}:

\begin{proof}[Proof of \Cref{thm: explicit}]
From \Cref{lem: diminishing}, under our stated assumptions diminishing adaptation is satisfied and the sequence $(\theta_n)_{n \geq 0}$ is contained in a compact set $\Theta_0$.
From \Cref{lem: explicit SSGE}, under our stated assumptions, the $\phi$-MH Markov transition kernels $\{P_\theta\}_{\theta \in \Theta_0}$ are \ac{ssage}.
The ergodicity of \ac{rlmh} then follows from \Cref{thm: roberts}.    
\end{proof}

\subsection{Auxiliary Lemmas}
\label{app: auxiliary}

The following auxiliary lemmas were used in the proofs of \Cref{lem: ssge for phi mh,thm: explicit}.
The first, \Cref{lem: super exp tails}, is a well-known result that sub-exponential distributions decay uniformly quickly in the tail:

\begin{lemma}[Tails of sub-exponential distributions]
\label{lem: super exp tails}
    Let $p \in \mathcal{P}(\mathbb{R}^d)$.
    Then for all $r > 0$ there exists $R > 0$ such that
    $$
        \frac{p(x + s n(x))}{p(x)} \leq \exp(-s r), \qquad \forall \|x\| \geq R, s \geq 0
    $$
    and $C_{p(x)} = \{y : p(y) = p(x)\} \cong \mathbb{S}^{d-1}$.
\end{lemma}
\begin{proof}
    See Sec. 4. of \citealp{jarner2000geometric}.
\end{proof}

The next technical lemma, \Cref{lem: quasi-symmetric}, is a simple but novel cocntribution of this work, establishing sufficient conditions for quasi-symmetry to be satisfied:

\begin{lemma}[Quasi-symmetry for $\phi$-MH]
    \label{lem: quasi-symmetric}
    Let $\Theta_0$ be compact.
    In the setting of \Cref{thm: explicit}, the proposal $\{q_\theta(\cdot|x) : x \in \mathbb{R}^d\}$ is quasi-symmetric on $\Theta_0$, meaning that
    \begin{align}
    \rho := \sup_{\theta \in \Theta_0} \sup_{x,y \in \mathbb{R}^d} \frac{q_\theta(y|x)}{q_\theta(x|y)} < \infty . \label{eq: quasi sym const}
    \end{align}
\end{lemma}
\begin{proof}
    For this proposal,
    $$
    \frac{q_\theta(y|x)}{q_\theta(x|y)} =  \exp\left( - \|y - \phi_\theta(x)\|_{1,\Sigma} + \|x - \phi_\theta(y)\|_{1,\Sigma} \right) .
    $$
    Let $z = x - y$ and $\eta_x^\theta  = x - \phi_\theta(x)$.
    Let $M := \sup_{\theta \in \Theta_0} \sup_{x \in \mathbb{R}^d} \|x - \phi_\theta(x) \|_{1,\Sigma}$, which exists by the assumed local boundedness of $\theta \mapsto \sup_{x \in \mathbb{R}^d} \|x - \phi_\theta(x) \|$, compactness of $\Theta_0$, and the norm-equivalence in \eqref{eq: norm equiv}.
    Then we have a bound
    \begin{align*}
        \left| - \|y - \phi_\theta(x)\|_{1,\Sigma} + \|x - \phi_\theta(y) \|_{1,\Sigma} \right| & = \left| \| \eta_y^\theta - z \|_{1,\Sigma} - \|\eta_x^\theta - z\|_{1,\Sigma}  \right| \\
        & \leq \|\eta_x^\theta - \eta_y^\theta\|_{1,\Sigma}
        \leq \|\eta_x^\theta\|_{1,\Sigma} + \|\eta_y^\theta\|_{1,\Sigma} \leq 2M .
    \end{align*}
    where we have used the reverse triangle inequality.
    The result is then established with $\rho \leq \exp(2M)$ being an explicit bound on the quasi-symmetry constant.
\end{proof}

The next technical lemma, \Cref{lem: geometry of Ax}, concerns the geometry of the acceptance set $A_\theta(x)$.
In the case of a random walk proposal, $\rho = 1$ and the first part of \Cref{lem: geometry of Ax} reduces to the existing \Cref{lem: super exp tails}.
A novel contribution of this paper is to study the geometry of the acceptance set in the case of a more general Metropolis--Hastings proposal.

\begin{lemma}[Geometry of $\partial A_\theta(x)$]
\label{lem: geometry of Ax}
    Let $p \in \mathcal{P}(\mathbb{R}^d)$ and $\Theta_0 \subset \Theta = \mathbb{R}^p$.
    Suppose that $(\theta,x) \mapsto \phi_\theta(x)$ is continuous and $B = \sup_{\theta \in \Theta_0} \mathrm{Lip}(\phi_\theta) < \infty$.
    Assume quasi-symmetry with parameter $\rho$.
    Then for all $x$ large enough, and all $\theta \in \Theta_0$, $\partial A_\theta(x) \cong \mathbb{S}^{d-1}$.
\end{lemma}
\begin{proof}
    From \Cref{lem: super exp tails}, the assumption that $p \in \mathcal{P}(\mathbb{R}^d)$ implies there exists $R_1 > 0$ such that, for all $\|x\| \geq R_1$, we have $C_{\rho p(x)} \cong \mathbb{S}^{d-1}$.
    From \eqref{eq: def subexponential}, there exists $R_2 \geq R_1$ sufficiently large that
    \begin{align}
        \sup_{\|x\| \geq R_2} n(x) \cdot \nabla \log p(x) < - \sqrt{d} \lambda_{\min}^{-1/2}(\Sigma) (1+B) . \label{eq: suppose R big}
    \end{align}
    From \Cref{lem: super exp tails} with $r = \log \rho$, and recalling the fact that $\rho \geq 1$, there exists $R_3 \geq R_2 + 1$ large enough that
    \begin{align*}
    \frac{p(x + s n(x))}{p(x)} \geq \frac{1}{\rho} \qquad \implies \qquad s \leq 1 
    \end{align*} 
    so that the radial distance between $C_{p(x)}$ and $C_{\rho p(x)}$ is uniformly at most 1 for all $\|x\| \geq R_3$.
    Finally, since $p(\cdot)$ is positive and continuous with $p(x) \rightarrow 0$ as $\|x\| \rightarrow \infty$ (from \Cref{lem: super exp tails}), there exists $R_4 \geq R_3$ such that for each $\|x\| \geq R_4$, we have $\|z\| \geq R_3$ for all $z \in C_{\rho p(x)}$.
    For $x,y \in \mathbb{R}^d$ with $\|x\| \geq R_4$, $y \neq 0$, and $p(y) \leq \rho p(x)$, let $r_{x,y}$ denote the (unique) positive constant such that $r_{x,y} n(y) \in C_{\rho p(x)}$.
    Then, for all $\|x\| \geq R_4$,
    \begin{align}
    \sup_{y : p(y) \leq \rho p(x)} \partial_s \log p(y + sn(y)) <  - \sqrt{d} \lambda_{\min}^{-1/2}(\Sigma) (1+B)  \label{eq: big gradient}
    \end{align}
    since for any $y$ with $p(y) \leq \rho p(x)$, we have that $\|y\| \geq \|r_{x,y} n(y)\| - 1 \geq R_3 - 1 \geq R_2$, and thus \eqref{eq: suppose R big} will hold.
    In the sequel we assume that $\|x\| \geq R_4$.

    From \eqref{eq: quasi sym A} the set $\partial A_\theta(x)$ is contained in the region bounded by $C_{\rho p(x)}$ and $C_{\rho^{-1} p(x)}$. 
    Let $\xi \in \mathbb{S}^{d-1}$ and consider the intersection of the set
    $\partial A_\theta(x)$ with the line segment $\gamma(\xi) = \{r \xi : r_{\min} \leq r \leq r_{\max}\}$ where $p(r_{\min} \xi) = \rho p(x)$ and $p(r_{\max} \xi) = \rho^{-1} p(x)$.
    Our first task is to establish that this intersection is a singleton set.
    
    Let $f_{x,\xi}(r) := \log p(r\xi) - \log p(x)$ and $g_{x,\xi,\theta}(r) := \log q_\theta(r\xi|x) - \log q_\theta(x|r\xi)$, so that we seek solutions to $f_{x,\xi}(r) = g_{x,\xi,\theta}(r)$ with $r \in [r_{\min},r_{\max}]$.
    Now $f_{x,\xi}$ is continuous with $f_{x,\xi}(r_{\min}) = \log \rho$ and $f_{x,\xi}(r_{\max}) = - \log \rho$.
    On the other hand, $g_{x,\xi,\theta} : [r_{\min},r_{\max}] \rightarrow \mathbb{R}$ is continuous and takes values only in $[- \log \rho , \log \rho]$, so the intersection $\partial A_\theta(x) \cap \gamma(\xi)$ is a non-empty set, and we can pick $r_\theta(\xi) \xi \in \partial A_\theta(x) \cap \gamma(\xi)$. 
    Since the region bounded by $C_{\rho p(x)}$ and $C_{\rho^{-1}p(x)}$ does not contain 0, it follows that $r_\theta(\xi) \in (0,\infty)$.

    Our next task to argue that $r_\theta(\xi)$ is the only element of this set.
    Now, 
    \begin{align*}
        g_{x,\xi,\theta}(r) & = \log q_\theta(r\xi|x) - \log q_\theta(x | r\xi) 
        = - \| r\xi - \phi_\theta(x) \|_{1,\Sigma} + \|x - \phi_\theta(r\xi) \|_{1,\Sigma} 
    \end{align*}
    from which it follows that $\mathrm{Lip}(g_{x,\xi,\theta}) \leq \sqrt{d} \lambda_{\min}^{-1/2}(\Sigma) (1+B)$, where we have used the norm-equivalence in \eqref{eq: norm equiv}.
    So $f_{x,\xi}(r)$ and $g_{x,\xi,\theta}(r)$ are equal at $r = r_\theta(\xi)$ and cannot be equal again on $r \in [r_{\min},r_{\max}]$ since from \eqref{eq: big gradient} the gradient of $f_{x,\xi}(r)$ is everywhere lower than the Lipschitz constant of $g_{x,\xi,\theta}(r)$ on $[r_{\min},r_{\max}]$.
    Thus $r = r_\theta(\xi)$ is the only solution to $f_{x,\xi}(r) = g_{x,\xi,\theta}(r)$.    

    Lastly we note that continuity of the map $\xi \mapsto r_\theta(\xi)$ follows from continuity of the maps $\xi \mapsto f_{x,\xi}$ and $\xi \mapsto g_{x,\xi,\theta}$, which in turn follows from the continuity of $x \mapsto p(x)$ and $(\theta,x) \mapsto \phi_\theta(x)$ that we assumed.
\end{proof}

The next technical lemma, \Cref{lem: lebesgue}, is a basic bound on the Lebesgue measure of the intersection of any set $C \cong \mathbb{S}^{d-1}$ with a Euclidean ball.
The proof closely follows an argument used within the proof of Theorem 4.1 in \citet{jarner2000geometric}, but we present it here to keep the paper self-contained:

\begin{lemma}
    \label{lem: lebesgue}
    Suppose that $C \cong \mathbb{S}^{d-1}$ and let $C^\epsilon := \{x + s n(x) : |s| \leq \epsilon \}$.
    Then, for all $R > 0$ and all $\|x\| > R$, 
    \begin{align*}
        \lambda_{\mathrm{Leb}}\left( C^\epsilon \cap B_{R}(x) \right) & \leq \epsilon \left( \frac{\|x\| + R}{\|x\| - R} \right)^{d-1} \frac{\lambda_{\mathrm{Leb}}(B_{3R}(x))}{R} .
    \end{align*}
\end{lemma}
\begin{proof}
    Recall that $C \cong \mathbb{S}^{d-1}$ means that we can parametrise $C$ as $C = \{r(\xi) \xi \; : \; \xi \in \mathbb{S}^{d-1}\}$ for some function $r : \mathbb{S}^{d-1} \rightarrow (0,\infty)$.
    Let $T(x) := \{ \xi \in \mathbb{S}^{d-1} : r \xi \in B_R(x) \; \text{for some} \; r \geq 0 \}$ and $S(x) := \{ r \xi : \xi \in T(x), \|x\| - R \leq r \leq \|x\| + R \}$.
    Then $B_R(x) \subset S(x) \subset B_{3R}(x)$.
    Let $\omega_d$ denote the surface measure on $\mathbb{S}^{d-1}$.
    The first inclusion leads to the bound
    \begin{align}
        \lambda_{\text{Leb}}(C^\epsilon \cap B_R(x)) = \int \mathrm{1}_{ C^\epsilon \cap B_R(x) }(y) \mathrm{d}y 
        & = \int_{T(x)} \left( \int_0^\infty \mathrm{1}_{C^\epsilon \cap B_R(x)}(r\xi) \; r^{d-1} \mathrm{d}r \right) \omega_d(\mathrm{d}\xi) \nonumber \\
        & \leq \int_{T(x)} \left( \int_{\|x\| - R}^{\|x\| + R} \mathrm{1}_{C^\epsilon}(r\xi) \; r^{d-1} \mathrm{d}r \right) \omega_d(\mathrm{d}\xi) \nonumber \\
        & \leq 2 \epsilon (\|x\| + R)^{d-1} \omega_d(T(x)) \label{eq: leb bound 1}
    \end{align}
    where for the final inequality we have used the fact that $C \cong \mathbb{S}^{d-1}$.
    The second inclusion leads to the bound
    \begin{align}
        \lambda_{\text{Leb}}(B_{3R}(x)) & \geq \lambda_{\text{Leb}}(S(x)) = \omega_d(T(x)) \int_{\|x\| - R}^{\|x\| + R} r^{d-1} \mathrm{d}r \nonumber \\
        & \geq \omega_d(T(x)) 2 R (\|x\| - R)^{d-1} . \label{eq: leb bound 2}
    \end{align}
    Combining \eqref{eq: leb bound 1} and \eqref{eq: leb bound 2} leads to the stated result.
\end{proof}

Our final auxiliary lemma is a standard property of distributions satisfying an interior cone condition, which appears in the standard convergence analysis of Metropolis--Hastings:

\begin{lemma}[Interior cone condition for $C_{p(x)}$]
\label{lem: cone}
    Let $p \in \mathcal{P}_0(\mathbb{R}^d)$ and fix $\epsilon > 0$.
    Then there exists $\delta > 0$ and $R > 0$ such that, for all $\|x\| \geq R$, the cones
    $$
    K_{\epsilon}(x) = \{ x - s \xi : 0 < s < \epsilon , \xi \in \mathbb{S}^{d-1} , \|\xi - n(x)\| \leq \delta / 2  \}
    $$
    lie in the interior of the compact region bounded by $C_{p(x)}$ and each have a Lebesgue measure $\zeta_\epsilon > 0$ that is $x$-independent.
\end{lemma}
\begin{proof}
    This is the content of the proof of Theorem 4.3 in \citet{jarner2000geometric}.
\end{proof}

\section{Implementation Detail}
\label{app: implementation}

This section explains how all algorithms referred to in the main text were implemented.
\Cref{subsec: ddpg} contains a brief introduction to \ac{ddpg}; an established approach to approximating the policy gradient.
\Cref{app: adaptive mcmc} contains full details for the \ac{arwmh} algorithm that was discussed in the main text.
\Cref{app: parametrisation} explains how our proposal was parametrised for \ac{rlmh} to ensure that the regularity conditions of \Cref{thm: explicit} were satisfied.
The specific implementational details that are required to reproduce our results, together with a discussion of the associated computational costs, are contained in \Cref{subsec: training detail}.

\subsection{Deep Deterministic Policy Gradient}
\label{subsec: ddpg}

This section describes the \ac{ddpg} algorithm for maximising $J(\phi_{\theta})$ based on the \emph{deterministic policy gradient theorem} of \cite{silver2014deterministic}.
Since the algorithm itself is quite detailed, we simply aim to landmark the main aspects of \ac{ddpg} and refer the reader to the original paper of \citet{lillicrap2015continuous} for further detail.

The deterministic policy gradient theorem states that
\begin{equation}
\label{eq:dpg_theorem}
\nabla_\theta J(\phi_\theta) = \mathbb{E}_{s\sim D_{\pi}}\left[\nabla_{\theta}\pi(s) \left .\nabla_{a}\matheuvm{Q}_{\pi}(s,a)  \right|_{a = \pi(s)}  \right],
\end{equation}
where the expectation here is taken with respect to the stationary distribution $s \sim D_{\pi}$ of the \ac{mdp} when the policy $\pi$ is fixed. 
The \emph{action-value function} $\matheuvm{Q}_{\pi}(s,a)$ gives the expected discounted cumulative reward from taking an action $a$ in state $s$ and following policy $\pi$ thereafter\footnote{Here, the notation for the action-value function $\matheuvm{Q}$ is not to be confused with the notation $Q$, which denotes the Metropolis--Hastings proposal distribution in the main text.}. 

Under \ac{ddpg}, the policy $\pi$ and the action-value function $\matheuvm{Q}$ are parameterised by neural networks whose parameters are updated via an \emph{actor-critic algorithm} based on \cite{silver2014deterministic}. Specifically, an \emph{actor network} $\pi_{\theta}(s)$ is updated in the direction of the policy gradient in \eqref{eq:dpg_theorem}, and a \emph{critic network} that approximates the action-value function, $\matheuvm{Q}_{w}(s,a) \approx \matheuvm{Q}_{\pi}(s,a)$, is updated by solving the \emph{Bellman equation} via stochastic approximation with off-policy samples from a \emph{replay buffer}. 
The Bellman equation is the name given to the recursive relationship
\begin{equation*}
\matheuvm{Q}_{\pi}(s_n, a_n) = \mathbb{E}_{s_{n+1} \sim D_{\pi}}\left[r_n + \gamma \matheuvm{Q}_{\pi}(s_{n+1}, \pi(s_{n+1})) \right],
\end{equation*}
and in \ac{ddpg} the critic network is trained by solving the optimisation problem
\begin{equation}
\label{eq:q_approx}
\argmin_{w} \mathbb{E}_{s_n \sim D_{\tilde{\pi}}, a_n \sim \tilde{\pi}} \left[(\matheuvm{Q}_{w}(s_n,a_n) - y_n)^2\right],
\end{equation}
where $y_n = r_n + \gamma \matheuvm{Q}_{w}(s_{n+1}, \pi(s_{n+1}))$, $a_n$ is generated from a stochastic \emph{behaviour policy} $\tilde{\pi}$, $D_{\tilde{\pi}}$ is the stationary state distribution according to $\tilde{\pi}$, and $s_{n+1}$ is resulted from interacting with the environment using $a_n \sim \tilde{\pi}(s_n)$. 
The expectation in \eqref{eq:q_approx} is approximated by sampling a mini-batch from a replay buffer that stores the experience tuples $\mathcal{R} := \{(s_n, a_n, r_n, s_{n+1})\}_{n=1}^{H}$. A common choice for $\tilde{\pi}$ is a noisy version of the deterministic policy $\pi$, e.g., $\tilde{\pi}(s) = \pi(s) + \varepsilon$, where $\varepsilon \sim \mathcal{N}(0, \Sigma_{\tilde{\pi}})$ with the covariance matrix $\Sigma_{\tilde{\pi}}$ that must be specified. 
Another popular choice is for $\varepsilon$ to follow an Ornstein–Uhlenbeck process. 
\cite{plappert2017parameter} noticed that \ac{ddpg} is still capable of learning successful policies even when run \emph{on-policy}, meaning that $\varepsilon = 0$. This relates to the fact that the replay buffer is naturally off-policy, by keeping experiences from policies that were previously visited. 
The implementation of \ac{ddpg} is summarised in Algorithm~\ref{alg: ddpg}, where $d_w$ and $d_{\theta}$ are determined by the architecture of the corresponding neural networks, and $\mathtt{Env}(\cdot)$ denotes a function that encapsulates the environment.

There are several aspects of \ac{ddpg} that are non-trivial, such as the distinction between \emph{target networks} $\matheuvm{Q}_{w'}$ and $\pi_{\theta'}$ and the current actor and critic networks $\matheuvm{Q}_w$ and $\pi_\theta$, and how these networks interact via the \emph{taming factor} $\tau$; see \citet{lillicrap2015continuous} for further detail.
For the purpose of this work we largely relied on default settings for \ac{ddpg} as implemented in \texttt{Matlab} R2024a; full details are contained in \Cref{subsec: training detail}.
The specific design of \ac{rl} methods for use in adaptive \ac{mcmc} was not explored, but might be an interesting avenue for future work.

\begin{algorithm}[h]
    \caption{DDPG; Algorithm 1 in \citealp{lillicrap2015continuous}}
    \label{alg: ddpg}
    \begin{algorithmic}
        \Require $s_1 \in \mathcal{S}$ (initial state), $w_1 \in \mathbb{R}^{d_{w}}$ (initial critic parameters), $\theta_1 \in \mathbb{R}^{d_{\theta}}$ (initial actor parameters), $T \in \mathbb{N}$ (number of iterations), $M \in \mathbb{N}$ (mini-batch size), $\gamma \in (0,1)$ (discount factor), $(\varsigma_i)_{i=1}^T \subset [0,\infty)$ (critic learning rate), $(\varrho_i)_{i=1}^T \subset [0,\infty)$ (actor learning rate), $\tau \in (0,1)$ (taming factor)

        \item \textbf{Initialise:} $w'_1 = w_1$, $\theta'_1 = \theta_1$, $\mathcal{R}_1 = \{\}$

        \For{$n=1$ \textbf{to} $T$}
            \State $a_n \sim \tilde{\pi}(s_n)$ \Comment{sample action from behaviour policy}
            \State $(r_n, s_{n+1}) \leftarrow \mathtt{Env}(a_n)$ \Comment{interact with environment}
            \State $\mathcal{R}_{n+1} \leftarrow \mathcal{R}_n \cup \{(s_n, a_n, r_n, s_{n+1})\}$ \Comment{append experience to replay buffer}
            \State $\{(s_{(i)}, a_{(i)}, r_{(i)}, s_{(i+1)})\}_{i=1}^M \sim \mathcal{R}_{n+1}$ \Comment{sample a mini-batch uniformly from replay buffer}
            \State $\{y_i \leftarrow r_{(i)} + \gamma \matheuvm{Q}_{w'}(s_{(i+1)}, \pi_{\theta'}(s_{(i+1)}))\}_{i=1}^M$ \Comment{compute target values using target networks}
            \State $w_{n+1} \leftarrow w_n - \varsigma_n \frac{1}{M}\sum_{i=1}^M \left. \nabla_{w} (\matheuvm{Q}_{w}(s_{(i)},a_{(i)}) - y_i)^2 \right|_{w = w_n}$  \Comment{update critic network}
            \State $\theta_{n+1} \leftarrow \theta_n + \varrho_n \frac{1}{M}\sum_{i=1}^M \!\!\left.\nabla_{\theta}\pi_{\theta}(s)\right|_{\substack{s=s_{(i)} \\ \theta=\theta_n}}\!\!\left.\nabla_{a}\matheuvm{Q}_{w}(s,a)\right|_{\substack{s=s_{(i)}, a=\pi(s_{(i)}) \\ w=w_{n+1}}}$ \Comment{update actor network}
            \State $w'_{n+1} \leftarrow \tau w_{n+1} + (1-\tau) w'_n$ \Comment{update target critic network}
            \State $\theta'_{n+1} \leftarrow \tau \theta_{n+1} + (1-\tau) \theta'_n$ \Comment{update target actor network}
        \EndFor
    \end{algorithmic}
\end{algorithm}

\subsection{Adaptive Metropolis--Hastings}
\label{app: adaptive mcmc}

This appendix describes the classical \ac{arwmh} algorithm was used to warm-start \ac{rlmh}, and as a comparator in the empirical assessment reported in \Cref{sec: empirical}.
The algorithm we used is formally called an \emph{adaptive Metropolis algorithm with global adaptive scaling} in \citet{andrieu2008tutorial}, and full pseudocode is presented in \Cref{alg: vamcmc}.

\begin{algorithm}[H]
    \caption{AMH; Algorithm 4 in \citealp{andrieu2008tutorial}}
    \label{alg: vamcmc}
    \begin{algorithmic}
        \Require $\alpha_\star \in (0,1)$ (target acceptance rate, here 0.234), $m \in \mathbb{N}$ (number of iterations), $(\gamma_i)_{i \geq 0} \subset [0,\infty)$ (learning rate)

        \item \textbf{Initialise:} $x_{0} = 0$, $\mu_{0} = 0$, $\Sigma_{0} = I$,  $\lambda_{0} = 1$

        \For{$i = 1$ \textbf{to} $m$}
            \State $x^{\star}_i \sim \mathcal{N}(x_{i-1} | \lambda_{i-1}\Sigma_{i-1})$ \Comment{propose next state}
            \State $\alpha_i \leftarrow \min(1, p(x^{\star}_i) / p(x_{i-1}) )$ \Comment{acceptance probability}
            \State $x_i \leftarrow x_i^{\star}$ with probability $\alpha_i$, else $x_{i} \leftarrow x_{i-1}$ \Comment{accept/reject}
            \State $\log(\lambda_i) \leftarrow  \log(\lambda_{i-1}) + \gamma_{i-1} ( \alpha_i - \alpha_\star  ) $ \Comment{refine proposal scale}
            \State $\mu_i \leftarrow \mu_{i-1} + \gamma_{i-1} (x_i - \mu_{i-1})$ \Comment{update mean approximation}
            \State $\Sigma_{i} \leftarrow \Sigma_{i-1} + \gamma_{i-1} \left[ (x_{i}-\mu_{i-1})(x_{i}-\mu_{i-1})^\top - \Sigma_{i-1} \right]$ \Comment{update covariance approximation}
        \EndFor
    \end{algorithmic}
\end{algorithm}

In brief, \Cref{alg: vamcmc} constructs a sequence of approximations $\Sigma_n$ to the covariance matrix of the target $p(\cdot)$, and then proposes new states $x_{n+1}^\star$ using a Gaussian distribution centred at the current state $x_n$ with covariance $\lambda_n \Sigma_n$.
The prefactor $\lambda_n$ is selected in such a manner that the proportion of accepted proposals aims to approach $0.234$, a value that is theoretically supported \citep{gelman1997weak,yang2020optimal}.

\begin{remark}[AMH as $\phi$-MH]
The \ac{arwmh} algorithm in \Cref{alg: vamcmc} can be viewed as an instance of $\phi$-MH with the building block proposals $q_\varphi(\cdot|x)$ being Gaussian with mean $x$ and covariance $\varphi \in \mathrm{S}_d^+$, where the algorithm attempts to learn a \emph{constant} function $\phi : \mathbb{R}^d \rightarrow \mathrm{S}_d^+$.
\end{remark}

For warm starting of \ac{rlmh} we performed $m = 10^4$ iterations of \ac{arwmh} to obtain samples $(x_i)_{i=-m+1}^0$, and we took the matrix $\Sigma$ appearing in \eqref{eq: Laplace proposal} to be the matrix $\Sigma_m$ obtained as the sample average of the final third of samples generates from \ac{arwmh}.
This approach enables us to exploit a rough approximation of the covariance of $p(\cdot)$, which is important in higher-dimensional problems, while removing the burden of simultaneously learning a proposal covariance, in addition to a proposal mean, in our set-up for \ac{rlmh}.

\subsection{Parametrisation of the Policy}
\label{app: parametrisation}

The aim of this appendix is to explain how the maps $\phi_\theta$ were parametrised for the experiments that we performed, and to explain how our choice of parametrisation ensured the conditions of \Cref{thm: explicit} were satisfied.

Let $\psi_\theta : \mathbb{R}^d \rightarrow \mathbb{R}^d$ be a collection of functions indexed by $\theta \in \mathbb{R}^p$ such that $(\theta,x) \mapsto \psi_\theta(x)$ is locally Lipschitz over $(\theta,x) \in \mathbb{R}^p \times \mathbb{R}^d$.
Let $C \subset \mathbb{R}^d$ be a compact set and let $\gamma_C : \mathbb{R}^d \rightarrow \mathbb{R}$ be a smooth function with $\gamma_C(x) = 1$ on $x \in C^{\mathtt{c}}$.
Armed with these tools, we propose to set
\begin{align}
\phi_\theta(x) := \psi_\theta(x) + \gamma_C(x) [ x - \psi_\theta(x) ] \label{eq: construction}
\end{align}
for all $x \in \mathbb{R}^d$ and $\theta \in \mathbb{R}^p$.
The construction in \eqref{eq: construction} ensures that the regularity conditions (i)-(iii) of \Cref{thm: explicit} are satisfied, so that the ergodicity of \ac{rlmh} can be guaranteed.
Intuitively, the proposal \eqref{eq: Laplace proposal} will default to a random walk proposal $\phi_\theta(x) = x$ when the state $x$ is outside of the set $C$, as a consequence $\gamma_C(x) = 1$ in \eqref{eq: construction}, while when $x \in C$ there is an opportunity to learn a flexible mean $\psi_\theta(x)$ for the proposal \eqref{eq: Laplace proposal}.
First we establish that the regularity requirements (i)-(iii) of \Cref{thm: explicit} are indeed satisfied, and then specific choices for $\psi_\theta$, $C$, and $\gamma_C$ will be presented.

\medskip

\noindent \textbf{Property (i):}
For any compact $\Theta_0 \subset \mathbb{R}^p$,
$$
\sup_{\theta \in \Theta_0} \sup_{x \in \mathbb{R}^d} \|x - \phi_\theta(x)\| \leq \sup_{\theta \in \Theta_0} \sup_{x \in C} \|x - \phi_\theta(x)\| < \infty ,
$$
where the first inequality holds since $x - \phi_\theta(x)$ vanishes when $x \in C^{\mathtt{c}}$, and the second inequality holds since $(\theta , x) \mapsto x - \phi_\theta(x)$ is continuous, and hence bounded on the compact set $\Theta_0 \times C$.

\noindent \textbf{Property (ii):}
For any compact $\Theta_0 \subset \mathbb{R}^p$,
\begin{align}
\sup_{\theta \in \Theta_0} \text{Lip}(x \mapsto \phi_\theta(x)) & = 
\sup_{\theta \in \Theta_0} \sup_{x \in \mathbb{R}^d} \text{LocLip}_x (y \mapsto \phi_\theta(y)) \nonumber \\
& \leq 1 + \sup_{\theta \in \Theta_0} \sup_{x \in C} \text{LocLip}_x (y \mapsto \phi_\theta(y)) \nonumber \\
& \leq 1 + \underbrace{ \sup_{\theta \in \Theta_0} \sup_{x \in C} \text{LocLip}_{(\theta,x)}( (\vartheta,y) \mapsto \phi_\vartheta(y) ) }_{(*)} <  \infty , \label{eq: lips ineq}
\end{align}
where the first equality is the definition of the Lipschitz constant, the first inequality holds since $\phi_\theta(x)$ is the identity on $x \in C^{\mathtt{c}}$ with unit Lipschitz constant, the second inequality holds via set inclusion, and the final inequality holds since $\Theta_0 \times C$ is compact and therefore $(\theta,x) \mapsto \phi_\theta(x)$ is Lipschitz when restricted to $\Theta_0 \times C$, with $(*)$ the corresponding Lipschitz constant.

\noindent \textbf{Property (iii):}
For any compact $\Theta_0 \subset \mathbb{R}^p$,
\begin{align*}
    \sup_{\theta \in \Theta_0} \sup_{x \in \mathbb{R}^d} \text{LocLip}_\theta (\vartheta \mapsto \phi_\vartheta(x)) & = \sup_{\theta \in \Theta_0} \sup_{x \in C} \text{LocLip}_x (y \mapsto \phi_\theta(y)) \\
    & \leq \sup_{\theta \in \Theta_0} \sup_{x \in C} \text{LocLip}_{(\theta,x)}( (\vartheta,y) \mapsto \phi_\vartheta(y) ) < \infty , 
\end{align*}
where the equality holds since $\phi_\theta(x)$ is constant in $\theta$ when $x \in C^{\mathtt{c}}$, the first inequality holds by set inclusion, and the final inequality was established in \eqref{eq: lips ineq}.

\medskip

To implement the construction in \eqref{eq: construction} we need to specify a parametric map $\psi_\theta : \mathbb{R}^d \rightarrow \mathbb{R}^d$, a compact set $C \subset \mathbb{R}^d$, and a smooth function $\gamma_C$ that vanishes on $C^{\mathtt{c}}$.
For the experiments that we report in this manuscript we took:
\begin{itemize}
    \item $\psi_\theta(x) = \bar{x} + \Sigma^{1/2} \nu_\theta(x)$, where $\bar{x}$ is the mean and $\Sigma \in \mathrm{S}_d^+$ is the covariance matrix obtained from the warm-up samples $(x_i)_{i=-m+1}^0$ as explained in \Cref{sec: empirical}, and $\nu_\theta : \mathbb{R}^d \rightarrow \mathbb{R}^d$ is a neural network whose architecture and initialisation are detailed in \Cref{subsec: training detail}. 
    \item the set $C$ was taken to be an ellipsoid 
    $$
    C = \{x \in \mathbb{R}^d : \eta(x) \leq 1 \}, \qquad \eta(x) := \frac{ \|\Sigma^{-1/2}(x-\bar{x})\|^2 }{ \ell^2 } ,
    $$
    where $\ell > 0$ is a radius to be specified.
    \item the map $\gamma_C$ was taken to be the smooth transition function $\gamma_C(x) = \gamma(\eta(x))$ where 
    $$
    \gamma(\eta) := \left\{ \begin{array}{ll} 0 & \eta \in [0,1/2] \\
    \left[ 1 + \exp\left( - \frac{ 4\eta-3 }{ 4\eta^2 - 6\eta + 2 } \right) \right]^{-1} & \eta \in [1/2,1] \\
    1 & \eta \in [1,\infty) \end{array} \right.
    $$
    which satisfies the smoothness requirement and is identically one when $x \in C^{\mathtt{c}}$.
\end{itemize}
The radius $\ell$ of the ellipsoid $C$ was set to $\ell = 10$, representing approximately 10 standard deviations from the mean of $p(\cdot)$, which ensures that the symmetric random walk behaviour (that occurs under \eqref{eq: construction} when $x \in C^{\mathtt{c}}$) is rarely encountered.

\subsection{Training Details}
\label{subsec: training detail}

This section contains the implementational details required to reproduce the experimental results reported in \Cref{sec: empirical}.

\paragraph{Parametrisation of $\phi_\theta$}

As explained in \Cref{app: parametrisation}, the function $\phi_\theta$ was parametrised using a flexible parametric map $\nu_\theta : \mathbb{R}^d \rightarrow \mathbb{R}^d$.
For our experiments we took $\nu_\theta$ to be a fully-connected two-layer neural network with the ReLU activation function and 32 features in the hidden layer; a total of $p = (32 + d) (d+1)$ parameters to be inferred.

\paragraph{Pre-training of $\phi_\theta$}

The parameters $\theta$ of the neural network $\nu_\theta$ were initialised by pre-training against the loss 
$$
\theta \mapsto \frac{1}{m} \sum_{i = 1}^m \left\| \Sigma^{-\frac{1}{2}} (\bar{x} - x_{i-m}) - \nu_\theta(x_{i-m}) \right\|^2 ,
$$ 
computed over the warm-up samples $(x_i)_{i=-m+1}^0$ generated from \ac{arwmh}, so that the proposal corresponding to $\nu_\theta$ approximates the anti-correlated behaviour illustrated in \Cref{ex: independence sampler} of the main text.

Optimisation was conducted using the Deep Learning toolbox in \texttt{Matlab} R2024a, with the ADAM optimisation method \citep{kingma2014adam}.
The default settings of the toolbox were employed.
Pre-training was terminated once either the mean squared error validation loss (computed using a held out subset of 30\% of the dataset) fell below the threshold `1', or a maximum of 2000 epochs was reached. 
Upon termination, the network with the best validation loss was returned.

\paragraph{Training of $\phi_\theta$}

\ac{rlmh} was performed using the implementation of \ac{ddpg} provided in the \ac{rl} toolbox of \texttt{Matlab} R2024a.
The default settings for training $\phi_\theta$ using this toolbox were employed with:
\begin{itemize}
    \item 100 episodes, each consisting of $500$ iterations of \ac{mcmc}
    \item standard deviation of the Ornstein-Uhlenbeck process noise is 0 (c.f. \Cref{rem: on policy})
    \item actor learning rate = $10^{-6}$
    \item experience buffer length = $\frac{d}{ \|\Sigma\|_{F}^{2}} \land 10^{-5}$ ,
\end{itemize}
where $\|\Sigma\|_F$ denotes the Frobenius norm of $\Sigma \in \mathrm{S}_d^+$.

\paragraph{Parametrisation of the Critic $\matheuvm{Q}$}

For all experiments we took $\matheuvm{Q} : \mathbb{R}^{2d} \times \mathbb{R}^{2d} \rightarrow \mathbb{R}$ to be a fully-connected two-layer neural network with the ReLU activation function and 8 features in the hidden layer; a total of $(8 + 4d) (4d+1)$ parameters to be inferred.

\paragraph{Training of the Critic $\matheuvm{Q}$}

The default settings for training the critic $\matheuvm{Q}$ using the \ac{rl} toolbox in \texttt{Matlab} R2024a were employed, except for the maximum size of the replay buffer which was set to be $10^6$; large enough to retain the full sample path of the Markov chain.

\paragraph{Computation}

All computation was performed on a desktop PC with a 12th generation Intel i9-12900F (24) @ 2.419GHz CPU, 32GB RAM, and NVIDIA GeForce RTX 3060 Ti GPU.
The (median) average time required to perform \ac{rlmh} on a single task from the \texttt{PosteriorDB} benchmark was 132 seconds.

Note that such specifications are not required to run the experiments that we report; in particular it is not required to have access a GPU.

\section{Additional Empirical Details and Results}
\label{app: full empirical}

The sensitivity of our experiments to the choice of the neural network architecture is examined in \Cref{app: choice of network}.
A selection of additional illustrations of \ac{rlmh} are presented in \Cref{app: additional illus}.
\Cref{subsec: mala} describes how \ac{mala} was implemented.
The performance measures that we used for assessment are precisely defined in \Cref{app: performance measures}.
Full results for \texttt{PosteriorDB} are contained in \Cref{subsec: full results}.

\subsection{Choice of Neural Network}
\label{app: choice of network}

The sophistication of modern \ac{rl} methodologies, such as \ac{ddpg}, means that in practice there are several design choices to be specified.
For the present work we largely relied on the default settings provided in the \ac{rl} toolbox of \texttt{Matlab} R2024a, but it is still necessary for us to select the neural architectures that are used.
The aim of this appendix is to briefly explore the consequences of varying the neural architecture for $\phi_\theta$ in the context of the simple example from \Cref{fig: illustration}, to understand the sensitivity of \ac{rlmh} to the choice of neural network.

Results are displayed in \Cref{fig: sensitivity to network}.
For these experiments all settings were identical to that of \Cref{fig: illustration}, with the exception of gradient clipping; since the number of parameters $\text{dim}(\theta)$ in the neural network $\phi_\theta$ depends on the architecture of the neural network, the gradient clipping threshold $\tau$ in \Cref{alg: rlmh} was adjusted accordingly.
These results broadly indicate an insensitivity to the architecture of the neural network used to construct the proposal mean $\phi_\theta$ in \ac{rlmh}.
Specifically, both narrower and wider architectures, and also deeper architectures, all led to the same global mode-hopping proposal reported in \Cref{fig: illustration} of the main text.

\begin{figure}[!ht]
    \centering
    \begin{subfigure}{0.95\textwidth}
    \includegraphics[width=\linewidth,clip,trim = 2.8cm 0cm 2.7cm 0cm]{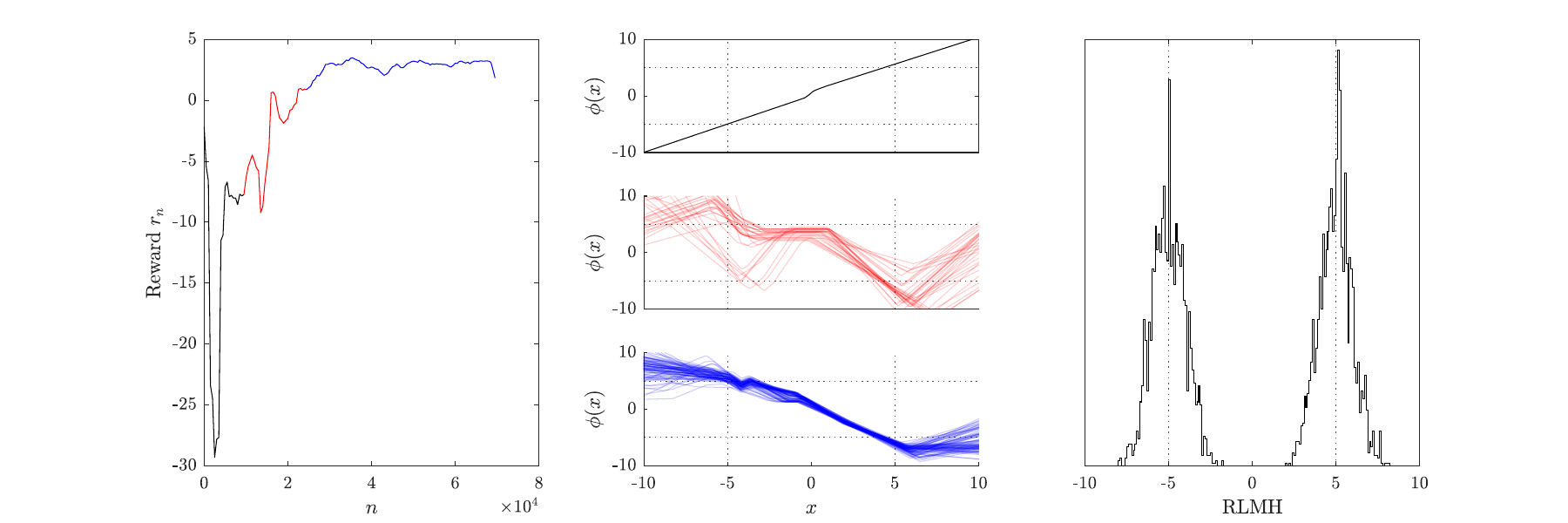}
    \subcaption{$h = 1$, $w = 16$}
    \end{subfigure}
    \begin{subfigure}{0.95\textwidth}
    \includegraphics[width=\linewidth,clip,trim = 2.8cm 0cm 2.7cm 0cm]{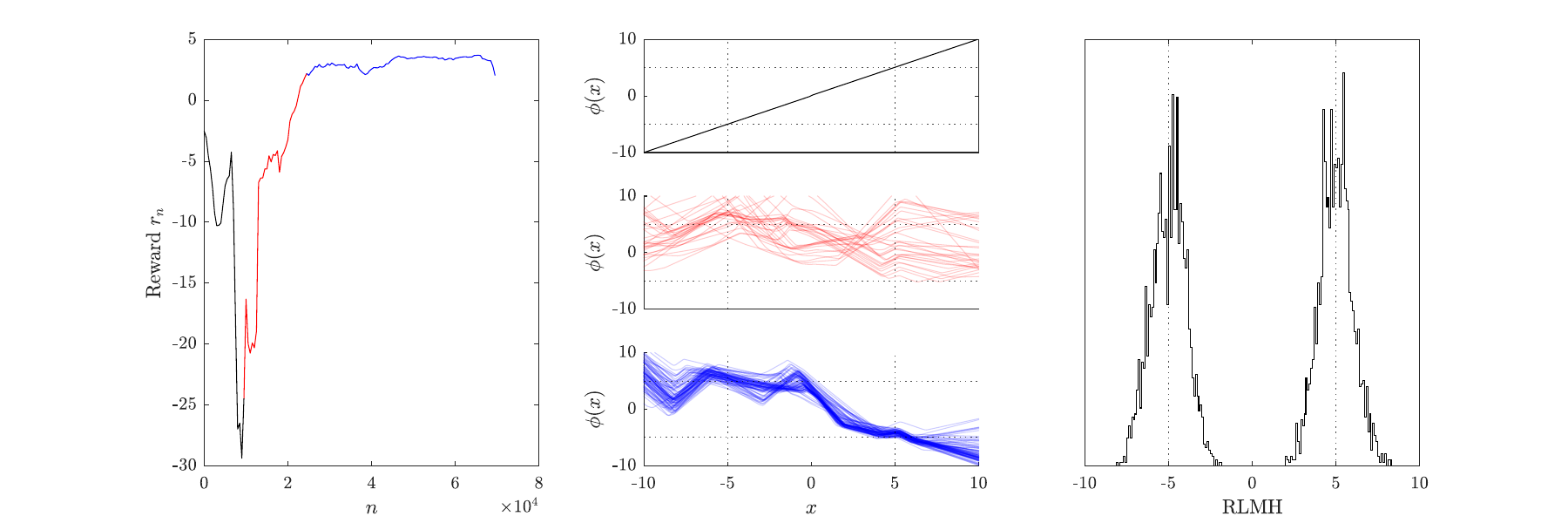}
    \subcaption{$h = 1$, $w = 64$}
    \end{subfigure}
    \caption{Investigating sensitivity to the architecture of the neural network $\phi$ in \ac{rlmh}.
    For the experiment presented in \Cref{fig: illustration} of the main text we employed a two layer (i.e. $h = 1$ hidden layer) neural network with width $w = 32$.
    The same experiment was performed with the architecture dimensions $(h,w)$ changed to (a) (1,16), (b) (1,64), (c) (1,256), (d) (2,32), and (e) (3,32); in all cases similar conclusions were obtained.
    [The colour convention and the interpretation of each panel is identical to that of \Cref{fig: illustration} in the main text.]}
    \label{fig: sensitivity to network}
\end{figure}

\begin{figure}[!ht] \ContinuedFloat
    \centering
    \begin{subfigure}{0.95\textwidth}
    \includegraphics[width=\linewidth,clip,trim = 2.8cm 0cm 2.7cm 0cm]{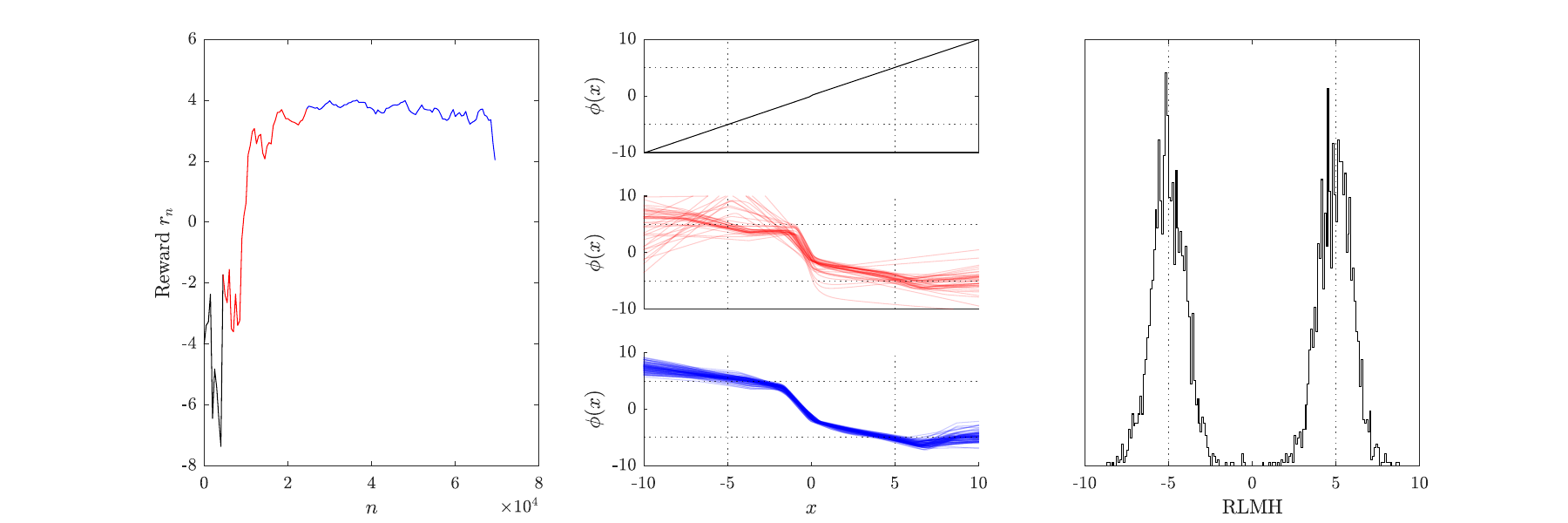}
    \subcaption{$h = 1$, $w = 256$}
    \end{subfigure}
    \begin{subfigure}{0.95\textwidth}
    \includegraphics[width=\linewidth,clip,trim = 2.8cm 0cm 2.7cm 0cm]{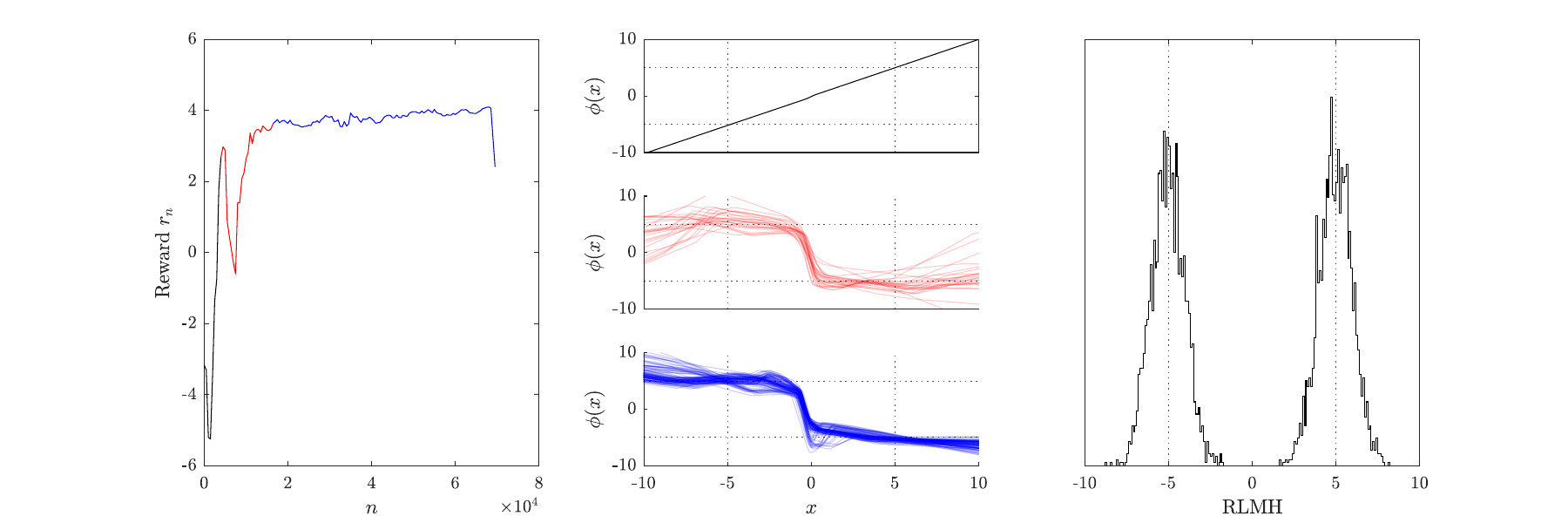}
    \subcaption{$h = 2$, $w = 32$}
    \end{subfigure}
    \begin{subfigure}{0.95\textwidth}
    \includegraphics[width=\linewidth,clip,trim = 2.8cm 0cm 2.7cm 0cm]{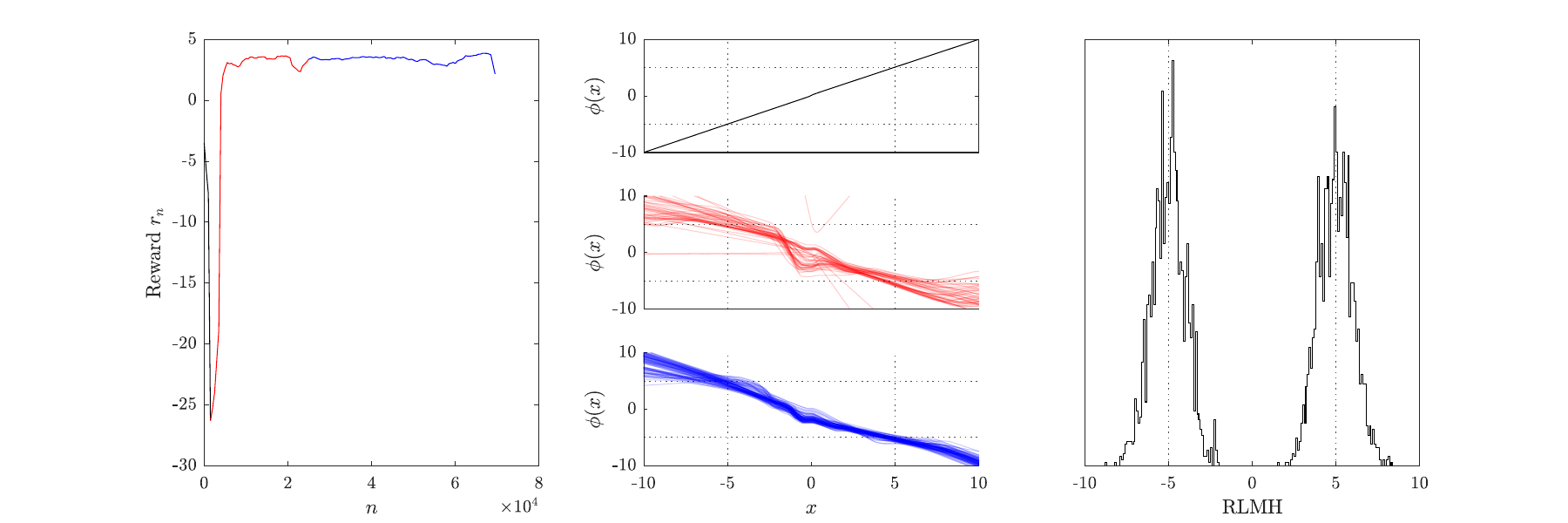}
    \subcaption{$h = 3$, $w = 32$}
    \end{subfigure} 
    \caption{Investigating sensitivity to the architecture of the neural network $\phi$ in \ac{rlmh}, continued.}
\end{figure}

\FloatBarrier

\subsection{Additional Illustrations in 1D and 2D}
\label{app: additional illus}

This appendix supplements \Cref{fig: illustration} in the main text with additional illustrations, corresponding to different target distributions $p(\cdot)$ in dimensions $d = 1$ and $d = 2$.
Specifically, we consider (a) a skewed target, (b) a skewed multimodal target, and (c) an unequally-weighted mixture model target in dimension $d = 1$, and a Gaussian mixture model target in dimension $d = 2$.
Results are reported in \Cref{fig: more illustrations} (for $d = 1$) and \Cref{fig: 2d illustration} (for $d = 2$).
These examples suggest that the gradient-free version of \ac{rlmh} can learn rapidly mixing Markov transition kernels for a range of different targets.

\begin{figure}[t!]
    \centering
    \begin{subfigure}{0.95\textwidth}
    \includegraphics[width=\linewidth,clip,trim = 2.8cm 0cm 2.7cm 0cm]{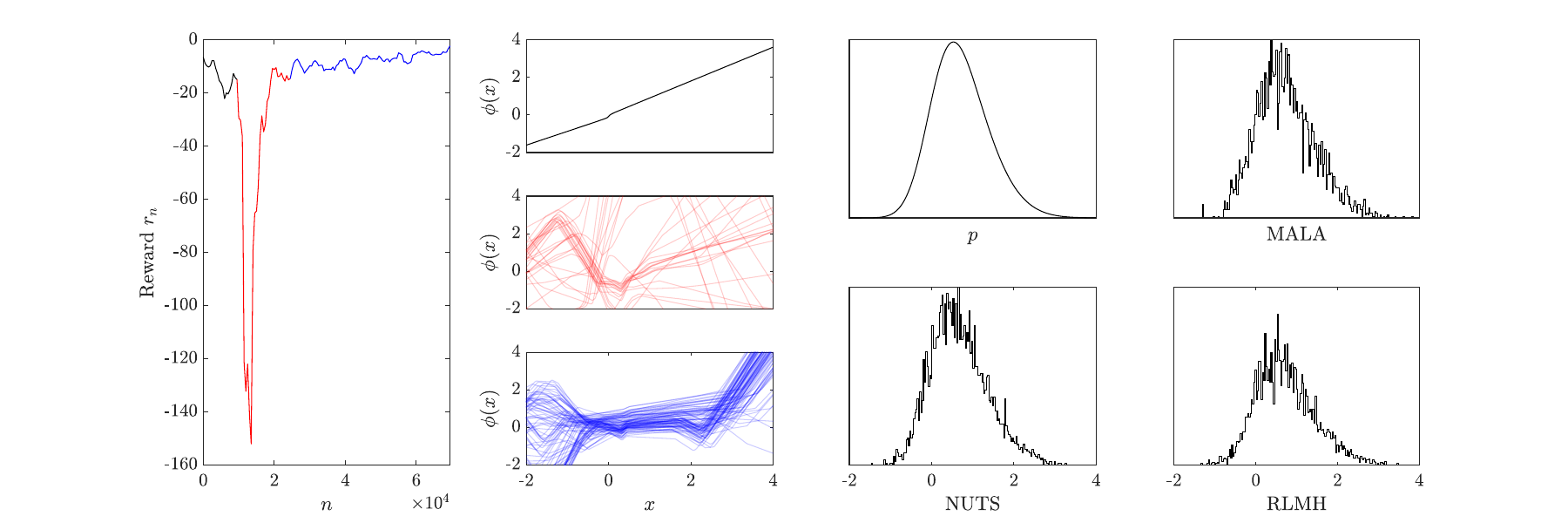}
    \subcaption{Skewed}
    \end{subfigure}

    \caption{\ac{rlmh}, illustrated.
    Here we considered (a) a skewed target, (b) a skewed multimodal target, and (c) an unequally-weighted mixture model target.
    Left: The reward sequence $(r_n)_{n \geq 0}$, where $r_n$ is the logarithm of the expected squared jump distance corresponding to iteration $n$ of \ac{rlmh}.
    Middle:  Proposal mean functions $x \mapsto \phi(x)$, at initialisation (top), and corresponding to the rewards indicated in red (middle) and blue (bottom).
    Right:  The density $p(\cdot)$, and histograms of the last $n = 5,000$ samples produced using \ac{mala}, \ac{nuts}, and \ac{rlmh}.
    [A smoothing window of length 5 was applied to the reward sequence to improve clarity of this plot.]
    \label{fig: more illustrations}
}
\end{figure}

\begin{figure}[t!] \ContinuedFloat
    \centering
    
     \begin{subfigure}{0.95\textwidth}
    \includegraphics[width=\linewidth,clip,trim = 2.8cm 0cm 2.7cm 0cm]{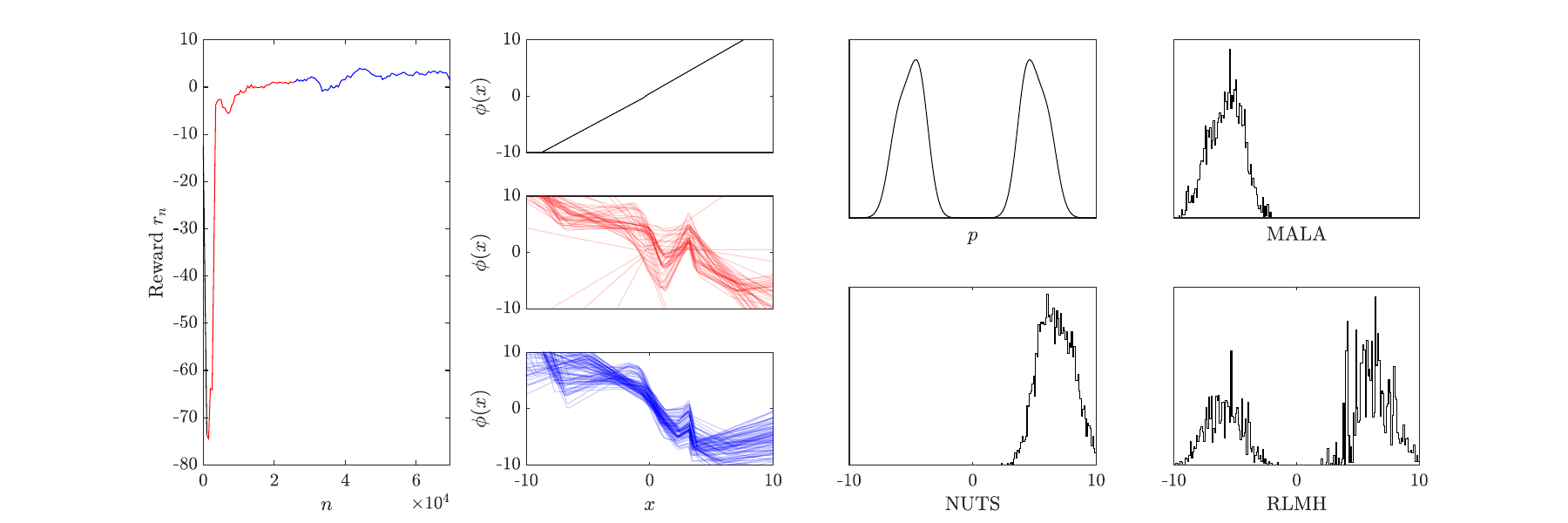}
    \subcaption{Skewed multimodal}
    \end{subfigure}

    \begin{subfigure}{0.95\textwidth}
    \includegraphics[width=\linewidth,clip,trim = 2.8cm 0cm 2.7cm 0cm]{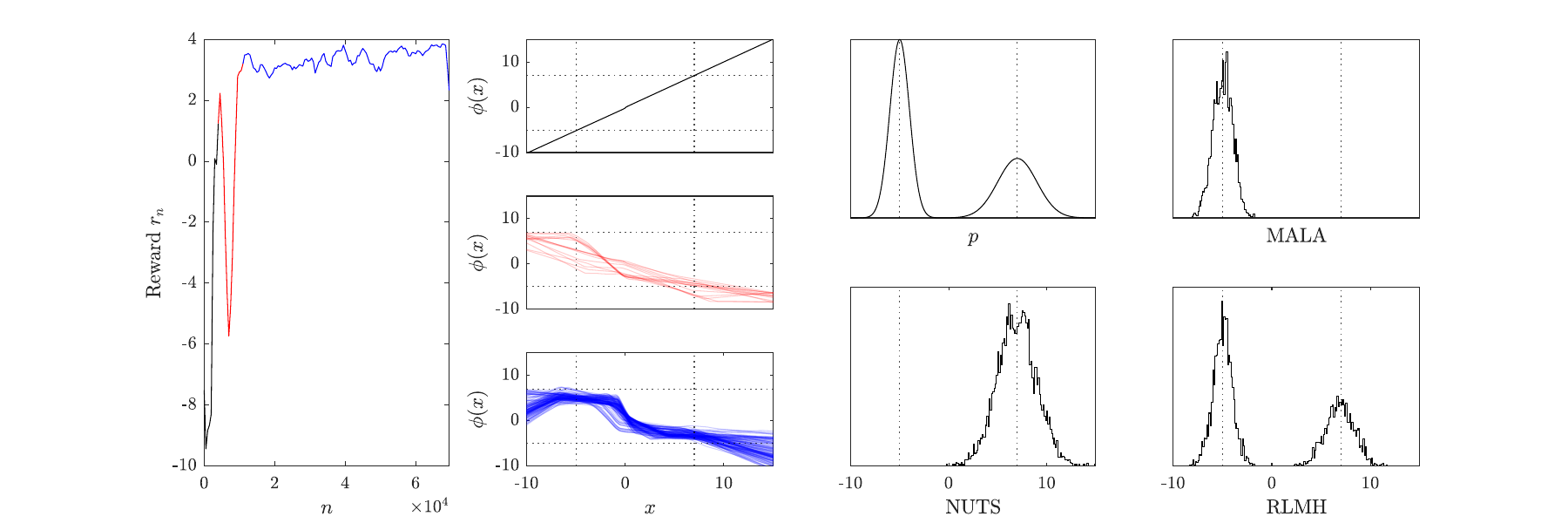}
    \subcaption{Unequal mixture model}
    \end{subfigure}

    \caption{\ac{rlmh}, illustrated, continued.
    Here we considered (a) a skewed target, (b) a skewed multimodal target, and (c) an unequally-weighted mixture model target.
    Left: The reward sequence $(r_n)_{n \geq 0}$, where $r_n$ is the logarithm of the expected squared jump distance corresponding to iteration $n$ of \ac{rlmh}.
    Middle:  Proposal mean functions $x \mapsto \phi(x)$, at initialisation (top), and corresponding to the rewards indicated in red (middle) and blue (bottom).
    Right:  The density $p(\cdot)$, and histograms of the last $n = 5,000$ samples produced using \ac{mala}, \ac{nuts}, and \ac{rlmh}.
    [A smoothing window of length 5 was applied to the reward sequence to improve clarity of this plot.]
    \label{fig: more illustrations}
}
\end{figure}

\begin{figure}[t!]
    \centering
    \includegraphics[width=\linewidth]{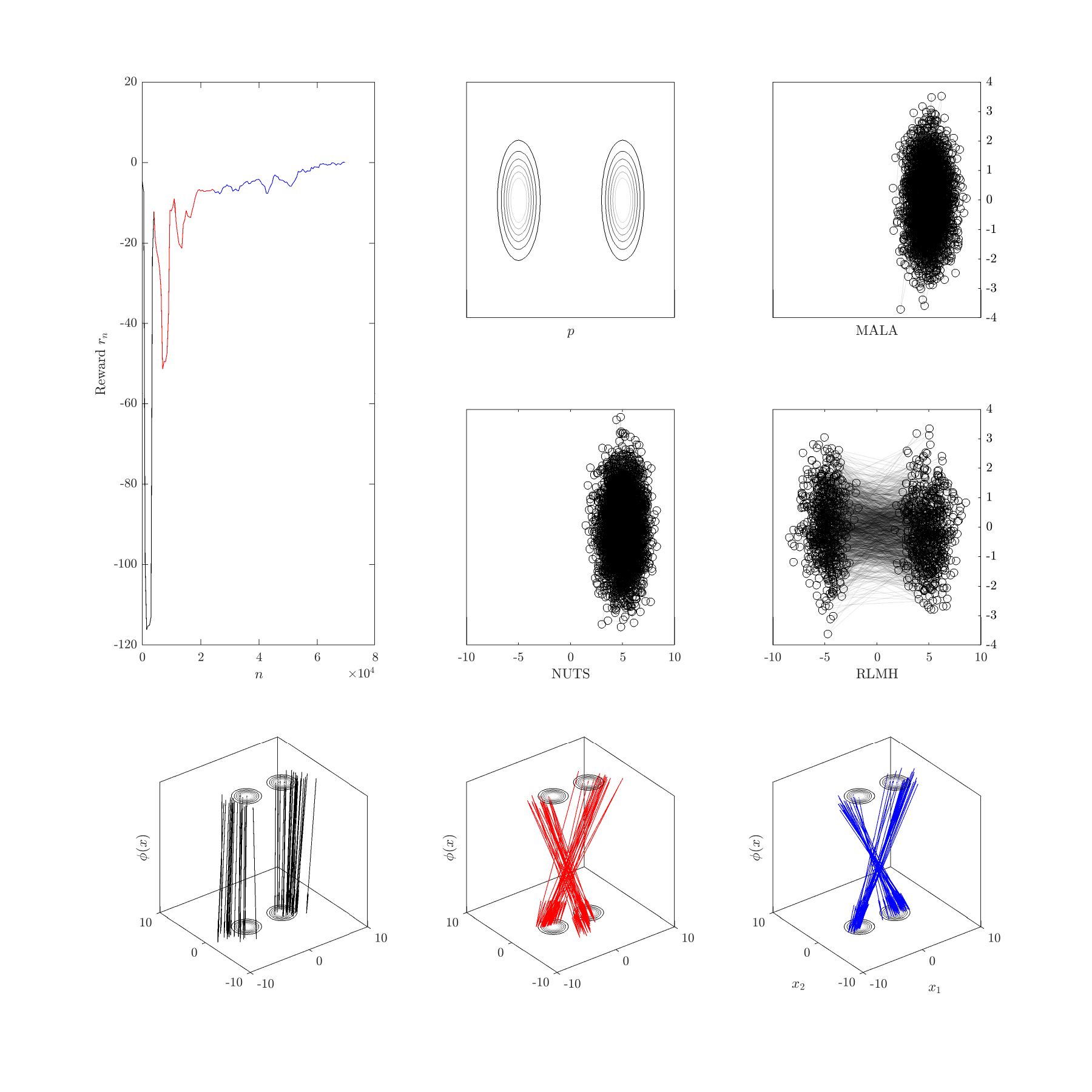}
    \caption{\ac{rlmh}, illustrated.
    Here we considered a two-dimensional Gaussian mixture model target.
    Top Left: The reward sequence $(r_n)_{n \geq 0}$, where $r_n$ is the logarithm of the expected squared jump distance corresponding to iteration $n$ of \ac{rlmh}.
    Top Right:  The density $p(\cdot)$, and histograms of the last $n = 5,000$ samples produced using \ac{mala}, \ac{nuts}, and \ac{rlmh}.
    Bottom:  Proposal mean functions $x \mapsto \phi(x)$, at initialisation (left), and corresponding to a typical policy from the period whose rewards are indicated in red (middle) and blue (right).
    [A smoothing window of length 5 was applied to the reward sequence to improve clarity of this plot.  In the bottom row we indicate the current state $x_n$ and the proposed state $x_{n+1}^\star$ of the Markov chain using a directed arrow from $x_n$ to $x_{n+1}^\star$.]
    \label{fig: 2d illustration}
    }
\end{figure}

\FloatBarrier

\FloatBarrier

\subsection{Adaptive Metropolis-Adjusted Langevin Algorithm}
\label{subsec: mala}

\Cref{alg: mala} contains pseudocode for an adaptive version of the (preconditioned) \acf{mala} algorithm of \citet{robert1996exponential}.
For the experiments that we report, we implemented this \ac{amala} in \texttt{Matlab} R2024a.

\begin{algorithm}[h]
    \caption{$\texttt{MALA}(x_0, \epsilon, \Sigma, n )$; \citealp{robert1996exponential}}
    \label{alg: mala}
    \begin{algorithmic}
        \Require $x_0 \in \mathbb{R}^d$ (initial state), $\epsilon > 0$ (scale of proposal), $\Sigma \in \mathrm{S}_d^+$ (preconditioner matrix), $n \in \mathbb{N}$ (number of iterations)
        \For{$i = 1$ \textbf{to} $n$}
            \State $x^{\star}_{i} \gets \underbrace{ x_{i-1} + \epsilon \Sigma (\nabla \log p)(x_{i-1}) }_{=: \nu(x_{i-1})} + (2 \epsilon \Sigma)^{1/2} Z_i$ \Comment{propose new state; $Z_i \sim \mathcal{N}(0,I)$}
            \State $\displaystyle \alpha_{i} \leftarrow \min\left(1, \frac{ p(x^{\star}_{i}) }{ p(x_{i-1}) } \frac{ \exp( - \frac{1}{4 \epsilon} \| \Sigma^{-1/2} ( x_{i-1} - \nu(x_i^\star) ) \|^2 ) }{ \exp( - \frac{1}{4 \epsilon} \| \Sigma^{-1/2} ( x_i^\star - \nu(x_{i-1}) ) \|^2 ) } \right)$ \Comment{acceptance probability}
            \State $x_{i} \leftarrow x_{i}^{\star}$ with probability $\alpha_{i}$, else $x_{i} \leftarrow x_{i-1}$ \Comment{accept/reject}
        \EndFor
        
        \State \textbf{Return:} $\{x_1,\dots,x_n\}$
    \end{algorithmic}
\end{algorithm}

For implementation of (non-adaptive) \ac{mala}, we are required to specify a step size $\epsilon > 0$ and a preconditioner matrix $\Sigma \in \mathrm{S}_d^+$ in \Cref{alg: mala}.
In general, suitable values for both of these parameters will be problem-dependent, and eliciting suitable values can be difficult \citep{livingstone2022barker}.
Standard practice is to perform some form of manual or automated tuning to arrive at parameter values for which the average acceptance rate is close to 0.574, motivated by the asymptotic analysis of \citet{roberts1998optimal}.
For the purpose of this work we implemented a particular adaptive version of \ac{mala} used in recent work such as \citet{wang2023stein}, which for completeness is described in \Cref{alg: ada_mala}.

\begin{algorithm}[h]
\begin{algorithmic}[1]
\Require $x_{0,0} \in \mathbb{R}^d$ (initial state), $\epsilon_0 > 0$ (initial scale of proposal), $\Sigma_0 \in \mathrm{S}_d^+$ (initial preconditioner matrix), $\{n_i\}_{i=0}^{h-1}$ (epoch lengths), $(\alpha_i)_{i=1}^{h-1} \subset (0 , \infty)$ (learning schedule), $E \in \mathbb{N}$ (number of epochs)

\State $\{x_{0,1} \ldots, x_{0,n_0}\} \leftarrow \texttt{MALA}(x_{0,0}, \epsilon_0, \Sigma_0, n_0)$
\For{$i=1,\dots,E-1$}
    \State $x_{i,0} \leftarrow x_{i-1, n_{i-1}}$
    \State $\rho_{i-1} \leftarrow \frac{1}{n_{i-1}} \sum_{j=1}^{n_{i-1}} 1_{x_{i-1,j} \neq x_{i-1,j-1}}$ \Comment{average acceptance rate for epoch $i$}
    \State $\epsilon_{i} \leftarrow \epsilon_{i-1} \exp(\rho_{i-1} - 0.574)$ \Comment{update scale of propsoal}
    \State $\Sigma_i \leftarrow \alpha_i \Sigma_i + (1 - \alpha_i) \mathrm{cov}(\{x_{i-1,1} \ldots, x_{i-1,n_{i-1}}\})$ \Comment{update preconditioner matrix}
    \State $\{x_{i,1} \ldots, x_{i,n_i}\} \leftarrow \texttt{MALA}(x_{i,0}, \epsilon_i, \Sigma_i, n_i)$
\EndFor 

\State \textbf{Return:} $\{x_{0,1},\dots,x_{E-1,n_{E-1}}\}$
\end{algorithmic}
\caption{Adaptive MALA}
\label{alg: ada_mala}
\end{algorithm}

For the pseudocode in \Cref{alg: ada_mala}, we use $\texttt{MALA}(x, \epsilon, \Sigma, n)$ to denote the output from \Cref{alg: mala}, and we use $\mathrm{cov}(\cdot)$ to denote the sample covariance matrix.
The algorithm monitors the average acceptance rate and increases or decreases it according to whether it is below or above, respectively, the 0.574 target.
For the preconditioner matrix, the sample covariance matrix of samples obtained from the previous run of \ac{mala} is used.
For all experiments that we report using \ac{amala}, we employed identical settings to those used in \citet{wang2023stein}.
Namely, we set $\epsilon_0=1$, $\Sigma_0=I_d$, $E=10$, and $\alpha_1 = \cdots = \alpha_9 = 0.3$.
The warm-up epoch lengths were $n_0 = \dots = n_8 = 1,000$ and the final epoch length was $n_9 = 10^5$.
The samples $\smash{ \{x_{E-1,1} , \dots , x_{E-1,n_{E-1}}\} }$ from the final epoch were returned, and constituted the output from \ac{amala} that was used for our experimental assessment.

\FloatBarrier

\subsection{Performance Measures}
\label{app: performance measures}

This section precisely defines the performance measures that were used as part of our assessment; \acf{esjd}, and \acf{mmd}.
For the comparison of adaptive \ac{mcmc} methods, we disabled adaptation after the initial training period in order to generate additional pairs $\{(x_{i-1},x_i^\star)\}_{i=1}^n$ with $n = 5,000$, from which the \ac{esjd} and \ac{mmd} were calculated.

\paragraph{Expected Squared Jump Distance}

The \ac{esjd} in each case was consistently estimated using
\begin{align*}
    \frac{1}{n} \sum_{i=1}^n \|x_i - x_{i-1} \|^2 ,
\end{align*}
the actual squared jump distance averaged over the sample path.
In principle a Rao--Blackwellised estimator could also be used, analogous to how the rewards $r_n$ are calculated in \ac{rlmh}, but we preferred to use the above simpler estimator as it generalises as a performance metric beyond Metropolis--Hastings \ac{mcmc}.

\paragraph{Maximum Mean Discrepancy}

The \ac{mmd} $D(P_m,Q_n)$ between a pair of empirical distributions $P_m = \frac{1}{m} \sum_{i=1}^m \delta_{x_i}$ and $Q_n = \frac{1}{n} \sum_{j=1}^n \delta_{y_j}$ is defined via the formula
\begin{align*}
    \mathrm{MMD}(P_m, Q_n)^2 & := \frac{1}{m^2} \sum_{i=1}^m \sum_{i' = 1}^m k(x_i,x_{i'}) - \frac{2}{mn} \sum_{i=1}^m \sum_{j=1}^n k(x_i, y_j) + \frac{1}{n^2} \sum_{j=1}^n \sum_{j'=1}^n k(y_j,y_{j'}) ,
\end{align*}
and for this work we took the kernel $k$ to be the Gaussian kernel
\begin{align*}
    k(x,y) & := \exp\left( - \frac{\|x - y\|^2}{\ell^2} \right)
\end{align*}
where the length-scale $\ell > 0$ was set according to the \emph{median heuristic}
\begin{align*}
    \ell & := \frac{1}{2} \mathrm{median}\{ \|y_i - y_j\| : 1 \leq i , j \leq n \}
\end{align*}
following \citet{garreau2017large}.
Here $P_m$ represents the approximation to the target $p(\cdot)$ produced using an adaptive \ac{mcmc} algorithm, and $Q_n$ represents a gold-standard set of $n = 10^4$ samples from the target, which are provided in \texttt{PosteriorDB}.

\subsection{Details for \texttt{PosteriorDB}}
\label{subsec: full results}

\texttt{PosteriorDB} is an attempt toward standardised benchmarking, consisting of a collection of posteriors to be numerically approximated \citep{magnusson2022posterior}. The test problems are defined in the \texttt{Stan} probabilistic programming language, and \texttt{BridgeStan} \citep{roualdes2023bridgestan} was used to directly access posterior densities and their gradients as required. 
At the time we conducted our research, \texttt{PosteriorDB} was at version 0.5.0 and contained 120 models, all of which came equipped with a gold-standard sample of size $n = 10^4$, generated from a long run of \ac{nuts}. 
Of these models, a subset of 44 were found to be compatible with \texttt{BridgeStan}, which was at version 2.4.1 at the time this research was performed. 
The version of \texttt{Stan} that we used was \texttt{Stanc3} version 2.34.0 (Unix). 
Thus we used a total of 44 test problems for our empirical assessment.

To implement \ac{arwmh} it is required to specify the learning rate $(\gamma_i)_{i \geq 0}$ appearing in \Cref{alg: vamcmc}.
Following Section 4.2.2 of \citet{andrieu2008tutorial}, we implemented a learning rate of the form
\begin{align}
\gamma_i = \frac{1}{ 2 \cdot (i+1)^{\beta} } \label{eq: simple learning rate}
\end{align}
where the exponent $\beta \in (0,1)$ was manually selected on a per-task basis to deliver the best performance for each of the 44 tasks from \texttt{PosteriorDB}.
Though it is possible to determine $\gamma_i$ at runtime, using techniques such as those described in \citet{delyon1993accelerated}, the simple schedule in \eqref{eq: simple learning rate} is most widely-used and performed reasonably well for most of the 44 tasks we considered.
The task-specific exponents $\beta$ that we used are included in the code used to produce these results, included as part of the electronic supplement.

\end{document}